\documentclass[10pt,reqno]{amsart}
\usepackage{todonotes}
\usepackage{commands}

\begin{document}

\title[Optimal Bubble Riding with Price-Dependent Entry]{Optimal Bubble Riding with Price-Dependent Entry:\\
a Mean Field Game of Controls with Common Noise}\blfootnote{The authors wish to thank Camilo Hern\'andez for fruitful discussions. L.T. is partially supported by the NSF grant DMS-2005832, the NSF CAREER award DMS-2143861 and the AMS Claytor-Gilmer fellowship.
S.W. is partially supported by the NSF grant DMS-2005832.}
\author{Ludovic Tangpi}
\author{Shichun Wang}
\maketitle
\begin{center}
\today
\end{center}
\begin{abstract}
In this paper we further extend the optimal bubble riding model proposed in \cite{TangpiWang22} by allowing for price-dependent entry times. Agents are characterized by their individual \emph{entry threshold} that represents their belief in the strength of the bubble. Conversely, the growth dynamics of the bubble is fueled by the influx of players.
Price-dependent entry naturally leads to a mean field game of controls with common noise and random entry time, for which we provide an existence result.
The equilibrium is obtained by first solving discretized versions of the game in the weak formulation and then examining the measurability property in the limit.
In this paper, the common noise comes from two sources: the price of the asset which all agents trade, and also the exogenous bubble burst time, which we also discretize and incorporate into the model via progressive enlargement of filtration. 
\end{abstract}

\section{Introduction}
Financial bubbles have become a topic of growing concern in the recent past. The classical view of \citet{Beaver68} that investors only follow ``information content'' is clearly not applicable to bubbles. Empirical evidence demonstrates the substantial stock price premium following vacuous company announcements of broad intentions to enter the cryptocurrency market \cite{AkyildirimCorbet20} and investors' overreaction to description of Blockchain activities in firms' 8-K disclosures \cite{RidingBlockchainMania19}. Journeying back another thirty years, a similar ``gold rush'' occurs during the dot-com bubble. Instead of avoiding the overpriced technology segment, sophisticated investors such as hedge funds invest heavily in the bubble while acknowledging an inevitable burst in the future \cite{HedgeFunds04, TechBubble2011}. The herding behavior is more prevalent now due to the increasing democratization of investing, as evidenced by the dramatic surge of retail traders during the ``meme-stock" frenzy \cite{CreditSuisse22}. Recent advancements in generative artificial intelligence (AI) unleash a frenzy both on Wall Street and among retail investors, pushing up stock prices of big technology companies. As NVIDIA's market capitalization marches across the trillion dollar line, many begin to suspect overvaluation in the semiconductor market. However, despite the unprecedented interest rate hikes and the recent turmoil in the cryptocurrency market, enthusiasm towards a potential AI bubble has not dampened. The intricate interplay between the intention to leverage rapid growth and the apprehension towards a future price adjustment provides the motivation for our equilibrium-based model proposed in \cite{TangpiWang22}, which we further investigate in this paper.

Substantial empirical evidence points to the inaccuracy of viewing bubbles merely as ``irrational exuberance''. A wealth of literature exists on the topic of rationality behind bubbles. The famous ``greater fool'' model formulated by \cite{ALLEN93}, as well as a more recent adaptation \cite{LIU2018}, pinpoints the driving factor behind bubble riding as the perception that others will acquire the overpriced asset in the future. \citet{AbreuBrunnermeier03} explore the idea of ``information asymmetry'' from another perspective by giving investors different entry times and various priors on the bubble formation time (see also \cite{RobustBubbleModel12, FiniteHorizon16} for extensions of this model). \citet{RidingBubbleConvex18} use a dynamic trading model to show that sophisticated, risk-averse money managers can invest in overvalued non-benchmark asset due to the presence of convex incentives. The authors in \cite{Rubinstein1987, AwayaIwasakiWatanabe22} argue that a ``chain of middlemen'' could also spur the escalation of the asset price. 

Despite extensive debates surrounding how a bubble is formed, or even defined, it is commonly agreed that the influx of investors and capital is what sustains the overvaluation. In other words, it is natural to model these events in a large population setting. This is the motivation behind the use of a mean field game (MFG) in our previous paper \cite{TangpiWang22}, which should be understood as the infinite population limit of symmetric stochastic differential games \cite{DelarueConv, TangpiConv, Lacker16, CardaliaguetDelarueLasryLions19}. First introduced by \citet{LasryLions1, LasryLions2, LasryLions33} and also by \citet{Huang1, Huang2}, mean field games provide tractable solutions compared to large but finite population games. We refer the readers to the monographs of \citet{CarmonaBookI, CarmonaBookII} for the probabilistic approach and to the notes of \citet{Cardaliaguet2012Notes} for the partial differential equation (PDE) approach to MFGs. 

Our previous study \cite{TangpiWang22} introduced a class of MFGs with varying entry times. Players begin to take advantage of the rising price trajectory at different times during the ``awareness window", a period viewed by \citet{AbreuBrunnermeier03} as a measure of heterogeneity among traders. The inflow of traders in turn fuels the price dynamics, whose drift is a function of the number of players currently in the game. We modeled the burst of the bubble as the minimum of exogenous and endogenous burst times. While an exogenous crash occurs due to events independent of trading, an endogenous crash happens when the average inventory of the players in the game falls below a threshold. We also included price impact as a second source of interactions among the agents through the controls, leading to an \emph{extended} MFG. We proved existence of MFG equilibrium using the method initiated in \cite{CarmonaLacker15}. Leveraging established methods on filtration enlargement, we were able to decompose the optimal strategy into before-and-after-burst segments, each part being progressively measurable with respect to the original filtration. Numerically, our model discovers that the equilibrium strategy attempts to delay the burst time and therefore sustain the growth if the bubble is large enough. 

The aim of this paper is to remove two major limitations of the model in \cite{TangpiWang22}. The first is that the entry times
could not depend on the price of the traded asset. They were instead modeled as independent samples from the same pre-determined distribution. However, an intuitive entry criterion for bubble riders is the first time the price crosses a certain threshold, which we use to characterize player influx in the present work. The second improvement is that we allow for an unbounded drift term in the price dynamics. In particular, since players' entry affects the price, the drift term will depend on the running maximum process of the price itself. We also provide a well-posedness result for this path-dependent dynamics.

As natural as these changes may seem, they require the model to incorporate a ``common noise'' which is famously challenging because the law of the population has to react to the realization of the noise. Just as in the case of a classical MFG, there are two approaches to deal with common noise. The analytic approach either reformulates the problem into a coupled system of stochastic PDEs or a deterministic, but infinite dimensional, PDE called \emph{the master equation}
(see \citet{CardaliaguetDelarueLasryLions19} for a careful presentation). Given a sufficiently smooth solution of the master equation, one can usually obtain strong results on the MFG equilibrium such as uniqueness or even regularity. However, almost all well-posedness results of the master equation require the Lasry-Lions monotonicity condition \cite{LasryLions33}, or the ``displacement monotonicity'' condition \cite{Ahuja16_WeakMonotone, DisplacementMono22, jackson2023quantitative} (see also the ``anti-monotonicity'' condition \cite{mou2022mean}). On the other hand, the probabilistic approach introduced by \citet{CarmonaCommonNoise} avoids making this assumption by a compactness argument. Since the monotonicity condition is too strong for our model, we take the latter route for constructing MFG equilibrium. A notable drawback to this compactness approach is that the controls might only be measurable with respect to a larger filtration. A well-known \emph{immersion} property is enforced to ensure fairness in observing that additional information. Immersion is a crucial property in the theory of filtration enlargement \cite{KharroubiLim14, PE2}, stochastic control \cite{Kurtz14, KarouiNguyenJeanblanc87}, the theory of conditional McKean-Vlasov SDEs \cite{Lacker2020SuperpositionAM} and of course mean field games \cite{CarmonaCommonNoise}. For an extensive discussion and generalization on both methods of tackling common noise, see \cite{CarmonaBookII}. Other recent extensions on related topics include MFGs with finite state space \cite{FiniteState21}, restoring uniqueness of equilibrium \cite{DelarueRestoreUnique19}, incorporating absorption \cite{Campi21}, convergence from N-player games \cite{DelarueConv, lacker2022closedloop}, and MFGs with interactions through controls \cite{djete2021extended}. To our knowledge, there aren't general existence results on the equilibrium of extended MFGs with common noise, which will be our main contribution in this work (Theorem \ref{Theorem:MFGExistence}). It is worth noting that although we provide a more intuitive model by incorporating common noise compared to our previous version in \cite{TangpiWang22}, the result is certainly weaker and less explicit, especially for numerical analysis.

The paper is structured as follows. In Section \ref{Section:ModelSetup}, we recall the features of the $N$-player model for bubble riding from \cite{TangpiWang22} and also introduce the new mechanism for price-dependent entry. Then we formulate the limit mean field game with common noise and state the assumptions as well as the main existence result. The proof is broken down into two parts. In Section \ref{Section:WeakControl} we first show an intermediate step with a weaker notion of admissibility. Then we strengthen the result in Section \ref{Section:StrongControl} by reducing the filtration for the controls. 

\section{Model Setup}\label{Section:ModelSetup}
\subsection{The N-Player Game}
\subsubsection{Price-dependent Entry}
Suppose the price trajectory $P_t$ starts at $P_0 > 0$ at time $t = 0$ when the bubble starts. Each player $i$ is characterized by $p^i\iid \nu_p$ on $[P_0, \infty)$, a price threshold for the player to deem the bubble worth riding. The player enters the game at $$t^i \ce \inf\{t \geq 0: P_{t} \geq p^i\} \wedge (T+1).$$
The value $T+1$ is chosen arbitrarily, but strictly greater than $T$, in order to bound $\T$ if the price never reaches the threshold. Note that in contrast to \cite{TangpiWang22}, the entry times are still random even if we fix the individual information $p^i$ due to their dependence on the common noise in $P$. If the price trajectory is c\`{a}dl\`{a}g and jumps are non-positive, then $P_{t^i} = p^i$ on the event that $t^i \leq T$. We assume that there are initial players with thresholds $P_0$ who are already in the game at $t = 0$. That is, $\nu_p(\{P_0\}) > 0$. Let $\vec{p} = (p^1, \dots, p^N)$ be the vector of thresholds. Then the number of players in the game by time $t$ is
\begin{equation}\label{eq:N_in}
N_{in}(t, P; \vec{p}) = \sumN \ind{t^i \leq t} = \sumN \ind{\max_{s\leq t} P_s \geq p^i } = NF_p^{N, \vec{p}}(\max_{s\leq t} P_s)
\end{equation}
where $F_p^{N, \vec{p}}$ is the empirical CDF of the thresholds $(p^1, \dots, p^N)$.

\subsubsection{Price Dynamics in the Bubble Phase}
Let $P^+$ denote the price process in the bubble phase before the burst. The growth of the bubble should depend on the inflow of players, $N_{in}(t, P; \vec{p})$, which by \eqref{eq:N_in} is a function of the maximum process of the bubble price itself. This motivates the following price dynamics before burst
\begin{equation}\label{eq:PriceDynamics+}
dP^+_t = b(t, \max_{s\leq t} P^+_s, P^+_t)dt + \sigma_0 dB_t, \quad P^+_0  = P_0,
\end{equation}
where $b$ is called the bubble trend function. Because the price grows with entry, $b$ depends on the thresholds $\vec{p}$ and should be monotone increasing in its second argument. We present the generalized form of the two examples for $b$ given in \cite{TangpiWang22} with price-dependent entry. 
\begin{example}[Exponential Bubble]\label{Example:1}
    \citet{AbreuBrunnermeier03} assumed a fully deterministic model with exponential price trajectory. In our setting, this translates to $$dP^+_t = \ell_t P^+_t dt + \sigma_0 dB_t, \quad \ell_t = \ell \frac{N_{in}(t, P;\vec{p})}{N} = \ell F_p^{N, \vec{p}}(\max_{s\leq t} P_s),$$
    where $\ell > 0$ stands for the peak growth rate of the bubble. If we assume that everyone enters at $t = 0$, that is having $p^i = P_0$ for all $i$, we obtain the model used in \cite{AbreuBrunnermeier03} with rate $\ell$.
\end{example}
\begin{example}[LPPL Bubble]\label{Example:2}
    The Johansen-Ledoit-Sornette (JLS) model proposed by \citet{JLS1} uses an assumption on the hazard rate $h_t$ of the burst time and arrives at a mean trajectory following the log-periodic power law (LPPL). While we model the burst time very differently, we can match the shape of the process by taking    $dP^+_t = h_t P^+_t dt + \sigma_0 dB_t$ with 
    \begin{equation*}
        h_t = A(t_c - t)^{\ell_t - 1} + C(t_c - t)^{\ell_t- 1}\cos(\omega \ln(t_c- t) - \phi),
    \end{equation*}
    with parameters $A, C, \omega, \phi$ and critical time $t_c$ set to $T$. In particular, the vanilla JLS model uses $\ell_t = \ell \in (0, 1)$ measures the power law acceleration of prices, which we generalize by incorporating the impact from the players' entry, namely $$\ell_t = \ell F_p^{N, \vec{p}}(\max_{s\leq t} P_s), \quad \ell \in (0, 1).$$
    Similar to the previous example, if all players enter at $t = 0$, the model reduces to a standard LPPL. See \cite{JLS2, JLS3, JLS4} for more detailed analysis of the JLS model.
\end{example}
It is not obvious that the path-dependent SDE \eqref{eq:PriceDynamics+} is strongly solvable, since $b$ is not necessarily bounded or Lipschitz, as in the case of empirical CDF. We defer the well-posedness result to the next section (Proposition \ref{Prop:PriceSDEWellPosed}), where monotonicity of the drift is the key condition that ensures existence and uniqueness of a strong solution. The risk of the bubble bursting is not reflected in \eqref{eq:PriceDynamics+} since this is the dynamics given that the bubble is still alive. We take a constant diffusion as the time horizon is relatively short-term by nature of a bubble. 

\subsubsection{Inventory and Trading Rate}
Each player $i$ has initial endowment $K^i_0 \iid \nu_K$ on $\R$ and goes ``all in'' when they enters the bubble ride. For simplicity, assume that there is no transaction cost when joining a bubble ride. That is, each player joins the game with $K^i_0/p^i$ shares of the bubble asset. We also allow for negative values of $K^i$, which represent a initial short (attack) position on the asset. Note that allowing price-dependent entry fixes a shortcoming of the original model proposed in \cite{TangpiWang22} where only the initial inventory is assumed to be i.i.d., which implies that the players who enter later will have more initial wealth due to the higher asset price at entry. 

Suppose that there is a common horizon $T > 0$. By choosing her trading rate $\alpha^i = (\alpha^i_t)_{t^i \leq t \leq T}$ after entry, the player can control her inventory trajectory by $$dX_t^i = \alpha_t^idt + \sigma dW^i_t,\quad X_t^i = 0 \text{ on } t<t^i, \quad X^i_{t^i} = K^i_0/p^i.$$ 
where $\sigma > 0$ is fixed and $W^i,\dots, W^N $ are independent $1$--dimensional Brownian motions corresponding to random streams of demand \cite{CarmonaLacker15,LealThesis}. A positive $\alpha_t$ corresponds to buying and a negative $\alpha_t$ corresponds to selling. We require $\alpha_t^i = 0$ on $t < t^i$ before entry for each $i \in \{1, \dots, N\}$, and $\alpha_t$ takes values in a compact interval $A \subset \R$. 

\subsubsection{Burst Time and Post-burst Price Dynamics}
Following \cite{AbreuBrunnermeier03} and \cite{TangpiWang22}, we allow the bubble to burst for both \emph{exogenous} and \emph{endogenous} reasons. An exogenous burst will be modeled as a non-anticipative random time, more specifically a totally inaccessible stopping time $\tau$ that is independent from the market information $(B, \{W^i\}_{i = 1, \dots, N})$. On the other hand, an endogenous burst occurs when the inventory of the active players (i.e. those who have entered) becomes too low to sustain the frenzy of the bubble. Define the empirical measure of the inventory $\mu^N_t$ and the average inventory $\bar{\mu}_t$ as 
\begin{equation}\label{empirical_meanmu}
\mu^N_t \ce \avg \sumN \delta_{X^i_t}, \quad \bar{\mu}_t^N \ce\frac{\ind{N_{in}(t, P, \vec{p}) \neq 0}}{N_{in}(t, P; \vec{p})} \sumN X_t^i  = \frac{\ind{N_{in}(t, P, \vec{p}) \neq 0}}{F^{N, \vec{p}}_{p}(\max_{s \leq t}P_s)}\int_\R x \mu_t^N(dx).
\end{equation}
Note that our assumption on the existence of initial players allows us to drop the indicator for all $t \in [0, T]$ when $N \to \infty$. 
For a given inventory threshold function $\zeta: [0, T] \to \R_+$, define the endogenous burst as
\begin{equation*}
\bar{\tau}^N(\mu^N) \ce \inf\cbra{t > \min_{i \in \{1, 2, \dots, N\}}t^i: \inf_{s \in [0, t]}\bar{\mu}^N_s \leq \zeta_t} \wedge T.
\end{equation*}
The true burst time is defined as the first of the two events: $$\tau^*(\mu^N) \ce \tau \wedge \bar\tau^N(\mu^N).$$
At burst time, the price drops by a fraction $\beta_t$ of the bubble component $\gamma_t$, defined by 
\begin{equation}\label{eq:bubble_component}
\gamma_t \ce \intt b(s, \max_{u\leq s} P^+_u, P^+_s)ds. 
\end{equation}
The function $\beta: [0,T] \to \R_+$ is referred to as the ``size'' of the bubble  \cite{AbreuBrunnermeier03}, or the ``loss amplitude'' in the JLS model \cite{JLS4}.

The setup becomes an optimal execution problem after the crash. Trades convey information which has a long-term impact to the price dynamics. When the bubble is present, the frenzy of the bubble growth dominates the impact from selling. However, after the crash, the asset price is governed by price impact within the short horizon. We use the game-theoretic extension of the model by \citet{AlmgrenandChriss01}, where the aggregate trading rate determines the instantaneous price impact. 

Let $\rho: A \to \R$ be a concave (hence also continuous) function for the instantaneous impact. See \cite{BOUCHAUD200957, MastromatteoTothBouchaud14, Amazing} for reasons of the concavity of price impact. 
Define the empirical measures of controls $\theta^N$ and the permanent price impact term $\qv{\rho, \theta^N_t}_{F_p^{N, \vec p}}$ as

\begin{equation}\label{empirical_meantheta}
\theta^N_t \ce \avg \sumN \delta_{\alpha^i_t}, \quad \qv{\rho, \theta^N_t}_{F_p^{N, \vec p}} \ce \frac{\ind{N_{in}(t, P, \vec{p}) \neq 0}}{N_{in}(t, P; \vec{p})}\sumN\rho(\alpha_t^i)  = \frac{\ind{N_{in}(t, P, \vec{p}) \neq 0}}{F^{N, \vec{p}}_{p}(\max_{s \leq t}P_s)}\int_A \rho(a) \theta^N_t(da).
\end{equation}
This is the second source of interaction among players currently in the game. Since by definition $\alpha^i$ and $X^i$ are both $0$ before entry, there is a factor $1/F^N_p$ in both \eqref{empirical_meanmu} and \eqref{empirical_meantheta} before the integral. After burst, the bubble trend is no longer present, so the price $P^-$ after burst follows on $[\tau^*, T]$
\begin{equation}\label{eq:PriceDynamics-}
dP^-_t = \qv{\rho, \theta^N_t}_{F_p^{N, \vec p}} + \sigma_0 dB_t, \quad P^{-}_{\tau^*} = P^+_{\tau*} - \beta_{\tau^*}\gamma_{\tau^*}.
\end{equation}
Define $D^*_t = \ind{t \geq \tau^*}$. Using chain rule on $P_t = P^+_t(1-D_t^*)  + P^-_tD_t^*$ gives us the price dynamics
\begin{equation}
    \begin{split}\label{eq:PriceDynamics}
            dP_t & = (1-D_t^*)dP^+_t - P_t^{+}dD_t^* + D_t^*dP^-_t + P^-_tdD_t^*\\
            & = \ind{t < \tau^*}dP^+_t + \ind{t \geq \tau^*}dP^-_t - \gamma_{\tau^*}\beta_{\tau^*}dD^*_t.
    \end{split}
\end{equation} 

\subsubsection{Objective and Equilibrium}
Player $i$'s cash process is modeled by $$dK^i_t = -\alpha^i_tP_t - \kappa(\alpha^i_t)dt, \quad K^i_0 \sim \nu_k.$$
where $\kappa(\cdot)$ is a continuous, strictly convex function satisfying $\kappa(0) = 0$ that measures the temporary price impact that affects only the individual trader and not the price itself. The well-known linear temporary impact \cite{AlmgrenandChriss01} corresponds to $\kappa$ being quadratic. See also \cite[Section 2.1]{CarmonaLacker15} for choosing $\kappa$ as the antiderivative of $\rho$. Note that the cash process remains at the initial endowment until the player enters the game, since $\alpha^i$ is kept at $0$. Under the usual self-financing condition, the pre-burst wealth $V^i$ of this player follows
\begin{align*}
    dV^i_t & = dK^i_t + X^i_tdP_t + P_tdX^i_t\\
    & = \pa{- \kappa(\alpha^i_t) + X_t^ib(t, \max_{s\leq t} P^+_s, P^+_t)\ind{t < \tau^*} +X_t^i\qv{\rho, \theta^N_t}_{F_p^{N, \vec p}}\ind{t \geq \tau^*} }dt\\
    & \qquad - X^i_t\beta_t\gamma_tdD^*_t + \sigma_0X_t^idB_t + \sigma P_t dW_t^i.
\end{align*}

The players are allowed to continue trading until $T$, even if the burst has already happened. But by definition of riding a bubble, the players do not believe in the fundamental value of the asset. Therefore, we impose a quadratic terminal inventory penalty $c(X^i_T)^2$ with $c > 0$ to encourage selling. For a fixed $\phi > 0$, we also impose a quadratic running inventory cost $\phi (X_t^i)^2$ which Cartea \textit{et al.} \cite{Carteaetal17} refer to as \emph{ambiguity aversion}.  Adding these costs to the negative of increase in wealth, we have the total cost of player $i$ that she wants to minimize: 
\begin{equation}\label{eq:Nobj}
	\begin{split}
		J^{N, i}(\balpha, \vec{p}) & \ce \E\sqbra{-(V_T^i - V_{t^i}^i) + \int_{t^i}^T\phi (X_t^i)^2dt + c(X_T^i)^2}  = \E\sqbra{c(X^i_T)^2 + X^i_{\tau^*}\beta_{\tau^*}\gamma_{\tau^*}}\\
  &+ \E\sqbra{\int_{t^i}^{T} \pa{\kappa(\alpha^i_t) + \phi (X^i_t)^2 - X_t^ib(t, \max_{s\leq t} P^+_s, P^+_t)\ind{t < \tau^*} + X_t^i\qv{\rho, \theta^N_t}_{F_p^{N, \vec p}}\ind{t \geq \tau^*}}dt}
	\end{split}
\end{equation}
for given vectors of strategies $\balpha = (\alpha^1, \dots, \alpha^N)$ and entry thresholds $\vec{p} = (p^1, \dots, p^N)$. The interaction among the players appears both in the price impact term through the average trading speed and also the burst time through the average inventory. We refer the readers to \cite{TangpiWang22} for more details on the model.

It is well-known that finite-player games of this type quickly becomes intractable as $N$ increases. Since the phenomenon of bubble riding fits naturally in the large-population setting, we shift our focus directly to the mean field limit of the game described above. 

\subsection{Mean Field Game Setup}
Let $(\Omega, \F, \P)$ be a probability space that supports independent $(K_0, W, B, \tau)$, whose law under $\P$ is $\nu_K \times \W \times \W \times \nu_\tau$ where $\W$ is the (one dimensional) Wiener measure. Let $\mathbb{F} = (\F_t)_{t \in [0, T]}$ be a $\P$-completed filtration defined on this probability space such that $W, B$ are $(\mathbb{F}, \P)$ Brownian motions, the initial wealth $K_0$ is $\F_0$-measurable, and exogenous burst time $\tau$ is not an $\mathbb{F}$-stopping time. Let $\mathbb{G}$ be the smallest filtration, on which $\tau$ is a stopping time, that contains $\mathbb{F}$. We will see in later sections that by construction, $\tau$ will in fact be a $\mathbb{G}$-totally inaccessible stopping time under mild assumptions. 

Let $\B(E)$ denote the Borel subsets of a Polish space $(E, d)$, and let $\mathcal{P}(E)$ denote the set of all probability measures on $\B(E)$. Unless specified otherwise, $\mathcal{P}(E)$ is equipped with the topology of weak convergence of measures, and $\mathcal{P}(E)$ is also a Polish space. Denote the Wasserstein space (or order 1) by $\calP_1(E)$, that is $$\calP_1(E) \ce \cbra{\mu \in \calP(E): \int_E d(x_0, x)\mu(dx) < \infty}.$$
where the choice of $x_0 \in E$ is arbitrary. Equip $\calP_1(E)$ with the 1-Wasserstein distance $$\W_1(\mu, \mu') \ce \inf_{\pi \in \Pi(\mu\times \mu')}\int_{E \times E} d(x, y)d\pi(x, y).$$
Let $D([0, T], \R)$ denote the space of all c\`adl\`ag functions from $[0, T]$ to $\R$. For a fixed $t^* \leq T$, define 
\begin{equation*}
    \X^{t^*} \ce \Big\{\bx^{t^*} \in D([0, T], \R): \bx_t^{t^*}= 0 \text{ on } [0,t^*), \text{ continuous on } [t^*, T]\Big\}\text{ and } \X^{*} \ce \bigcup_{t^* \in [0, T]} \X^{t^*}.
\end{equation*}
For each $\bx^{t^*} \in \X^*$, we require $t^*$ to be the largest value such that $\bx^{t^*} \in \X^{t^*}$ to avoid redundancies. Suppose $t_1, t_2 \in {[0, T]}$. Let $\bx^{t_1}, \by^{t_2} \in \X^{*}$. Notice that the standard sup distance $d(\bx^{t_1}, \by^{t_2} )= ||\bx^{t_1} - \by^{t_2}||_\infty$ is no longer suitable on $\X^*$ because it does not allow two processes to be close unless $t_1 = t_2$. Therefore, for each $\bx^{t^*} \in \X^*$, we can define its continuous counterpart $\bar{\bx}^{t^*} \in C([0, T], \R)$ as
 $$\bar{\bx}_t^{t^*} \ce \begin{cases}
    \bx_{t^*}^{t^*} & t \in [0, t^*)\\
    \bx_t^{t^*} & t \geq t^*.
\end{cases}
$$
We then define a metric on $\X^* \subset D([0, T], \R)$ to be 
\begin{equation}\label{supmetric}
    d_{\X^*}(\bx^{t_1}, \by^{t_2}) \ce \|\bar{\bx}^{t_1} - \bar{\by}^{t_2}\|_\infty + |t_1 - t_2|, \quad \text{and} \quad \norm{\bx^{t^*}}_{\X^*} \ce \norm{\bar{\bx}^{t^*}}_{\infty} + t^*.
\end{equation}

Define $\Theta^{A} \ce \calP([0, T] \times A)$ with time marginal being Lebesgue, and define
\begin{equation}\label{Definition:Mnew}
    \M \ce \cbra{\mu \in \mathcal{P}(\X^*): \int_{\X^*} \norm{\bx}_{\X^*}\mu(d\bx) < \infty} = \calP_1(\X^*).
\end{equation} 
Equip $\M$ and $\Theta^A$ with the topology of 1-Wasserstein convergence and weak convergence, respectively. Each $\theta\in \Theta^A$ can uniquely disintegrate into $\theta(dt, da) = \theta_t(da)$ with some measurable map $t \mapsto \theta_t \in \calP(A)$. Let $F_p$ be the CDF of the price threshold distribution $\nu_p$. Each $\mu \in \M$ can also be viewed as a $\calP(\R)$-valued process $t \mapsto \mu_t$, where $\mu_t = \mu \circ \calE_t^{-1}$ with $\calE_t: \X^* \to \R$ being the time coordinate mapping. Let $\mu$ and $\theta$ be the law of $X$ and $\alpha$. Then the average inventory among the players in the game and their price impact are the natural limit of \eqref{empirical_meanmu} and \eqref{empirical_meantheta}, namely:
$$\bar\mu_t \ce \frac{1}{F_p(\max_{s \leq t} P_s)}\qv{x, \mu_t}, \quad \qv{\rho, \theta_t}_{F_p} \ce \frac{1}{F_p(\max_{s \leq t} P_s)}\qv{\rho, \theta_t}.$$
\subsubsection{Burst Time}
While the exogenous burst time is the same as the $N$-player game, the endogenous burst now depends on $\bar\mu$. To ensure that there are already players in the game at time $t = 0$, we assume that $\nu_p(\{P_0\}) > 0$ and define $$\bar\tau(\mu) \ce \inf\cbra{t \in [0, T]: \inf_{s\in[0, t]}\bar\mu_s \leq \zeta_t} \wedge T,$$
where $\zeta:[0,T] \to \R^+$ is deterministic, continuous and strictly increasing with $\zeta_0 \in (0, \E[K_0]/P_0)$. The upper bound is set so that the bubble at least survives the initial players. These conditions guarantee enough regularity of $\bar\tau$ for the equilibrium to exist (see \cite[Sections 2.1.5 and 6.1]{TangpiWang22}). The actual burst time is again  $\tau^*(\mu) = \bar\tau(\mu) \wedge \tau$. Throughout the paper, we work under the following assumptions.
\begin{assumption}{E}\label{Assumption:E}~
\begin{enumerate}[label={(E\arabic*)}]
    \item \label{Assumption:E1}$A$ is a compact interval that includes $0$. 
    \item \label{Assumption:E2}$\nu_\tau$ is absolutely continuous with respect to Lebesgue measure on $\R^+$ satisfying $\nu_\tau({\tau > T}) > 0$. Moreover, its deterministic, non-negative intensity process $k: \R_+ \to \R_+$ is bounded by some $C_k > 0$ on $[0, T]$. 
    \item \label{Assumption:E3} $K_0$ has all moments. $\rho: A \to \R$ is locally bounded.
    \end{enumerate}
\end{assumption}
For a c\`{a}dl\`{a}g process $Y$, define $M^Y_t = \sup_{0 \leq s \leq t}Y_s$. Observe from Examples \ref{Example:1} and \ref{Example:2} that the bubble trend function $b$ depends on the running maximum of the price $P$ naturally through the CDF of entry thresholds. That is, the dynamics of $P$ depends on $F_p(M^P_t)$ where $F_p$ is not necessarily Lipschitz continuous. Therefore, the SDE defining the price process may not be well-posed. We will show in the appendix that appropriate growth and monotonicity conditions on $b$ below, which we make as standing assumptions on the bubble, allow to obtain well-posedness.
\begin{assumption}{B}\label{Assumption:B}~
\begin{enumerate}[label={(B\arabic*)}]
    \item \label{Assumption:B1}The bubble function $b: [0, T] \times [0, 1] \times \R \to \R$ is non-negative and satisfies the assumptions in Proposition \ref{Prop:PriceSDEWellPosed}.
    \item \label{Assumption:B2}The bubble size $\beta: [0, T] \times \Omega \to \R$ is a positive, continuous, bounded $\mathbb{F}$-progressively measurable process.
\end{enumerate}
\end{assumption}
Using Proposition \ref{Prop:PriceSDEWellPosed}, the price dynamics \eqref{eq:PriceDynamics} is well defined for a fixed $(\mu, \theta) \in \M \times \Theta^A$, namely
$$P_t = \ind{t < \tau^*)}P^+_t + \ind{t \geq \tau^*}P^-_t,$$
where the pre-burst price $P^+$ follows \eqref{eq:PriceDynamics+} and the post-burst price $P^-$ follows
\begin{align*}
    P_t^- &= P_{\tau^*}^+ - \beta_{\tau^*}\gamma_{\tau^*} + \int_{\tau^*}^t
\qv{\rho, \theta_s}_{F_p}ds + \sigma_0(B_t - B_{\tau^*})\\
& = P_0 + \gamma_{\tau^*} + \sigma_0B_{\tau^*} - \beta_{\tau^*}\gamma_{\tau^*} + \int_{\tau^*}^t
\qv{\rho, \theta_s}_{F_p}ds + \sigma_0(B_t - B_{\tau^*})\\
& = P_0 + \int_{\tau^*}^t
\qv{\rho, \theta_s}_{F_p}ds + \sigma_0B_t + (1-\beta_{\tau^*})\gamma_{\tau^*}.
\end{align*}
The bubble component $\gamma$ is defined by \eqref{eq:bubble_component}. Before crash, the bubble component dominates the drift of the price dynamics, whereas the impact term takes over after the the crash. Note that the price has exactly one jump at $\tau^*$, and the jump size $-\beta_{\tau^*}\gamma_{\tau^*}$ is always negative.
\subsubsection{Entry Time}\label{Subsub:EntryTime}
Since the underlying asset starts at a known value $P_0$, the entry threshold should be at least this amount, that is, $\nu_p$ is a distribution on $[P_0, \infty)$. Consider the product probability space $(\Omega \times [P_0, \infty), \F \otimes \B([P_0, \infty)), \P \otimes\nu_p)$. We naturally extend $(K_0, W, B, \tau)$ on the product space. The representative player enters the game at a random $\mathbb{G}$-stopping time $\T(p^*)$, $p^* \in [P_0, \infty)$, where $\T: \Omega \times [P_0, \infty) \to [0, T] \cup \{T + 1\}$ is given by
\begin{equation}\label{eq:EntryTime}
\T(p^*) \ce \inf\{t \in [0, T]: P_t \geq p^*\}\wedge (T+1), \quad \P-\as
\end{equation}
The value $T+1$ is again arbitrarily chosen. Then $\T(p^*)$ is a bounded $\mathbb{F}^B$ stopping time for every $p^* \in [P_0, \infty)$.

\begin{lemma}\label{Lemma:p*measurable}
    Consider $\T: \Omega \times [P_0, \infty) \to [0, T] \cup \{T + 1\}$ in equation \eqref{eq:EntryTime}. For $\P$-almost every $\omega$, the function $\T(\omega, \cdot)$ is strictly increasing until and if it reaches $T+1$. Moreover, it is left continuous with right limit on $[P_0, \infty)$, and $\T$ is jointly measurable.
\end{lemma}
\begin{proof}
    Monotonicity is obvious. To see that it is strict, we can first write $\T(p^*) = \inf\{t\geq 0: M^P_t \geq p^*\}\wedge (T+ 1)$. Note that the price process $P$ is $\P$-almost surely continuous except at $\tau^*$ where there is a non-positive jump. Therefore, $M^P$ is a monotone increasing, $\P$-almost surely continuous process, which implies strict monotonicity of $\T$.

    For each $\omega$ such that $M^P$ is continuous, $\T(p^*) = T+1$ for $p^* > M^P_T$. For $p^* \leq M^P_T$, we have $\T(p^*) \in [0, T]$ and $M_{\T(p^*)} = p^*$. Take an increasing sequence $p_n \uparrow p^*$. Then $\T(p_n)$ is also an increasing sequence which converges to some $t \leq \T(p^*)$ as $n \to \infty$. If $t < \T(p^*)$, then we can find $t' \in (t, \T(p^*))$ such that $p_n \leq M^P_{t'} \leq p^*$ for all $n$ by monotonicity of $M^P$. $p_n$ converging to $p^*$ implies $M^P_{t'} = p^*$ which contradicts the definition of $\T(p^*)$. The existence of right limit follows a similar argument. Joint measurability  follows from Lemma \ref{Lemma:Caratheodory}.
\end{proof}
\subsubsection{Admissibility of Controls}
By continuity of $P$, given $\T = t^*$ we can also recover the price threshold by $p^* = P_{t^*}$. However, it is still useful to define admissibility of controls in two separate ways: one in terms of entry times, and the other in terms of entry thresholds. To simplify notation, we denote by $\mathcal{PM}(\mathbb{G})$ (resp. $\mathbb{F}$) the $\sigma$-algebra generated by the $\mathbb{G}$ (resp. $\mathbb{F}$)-progressively measurable subsets of $\Omega \times [0, T]$.
\begin{definition}
Define the following sets for admissible controls:
\begin{itemize}
    \item For $t^* \in [0, T]$, let $\A({t^*})$ denote the set of square integrable, $\mathcal{PM}(\mathbb{G})$-measurable processes $\alpha: \Omega \times [0, T] \to A$ such that $\alpha_t  = 0$ for $t \in [0, t^*)$. We also set $\A(T+1)$ to be the singleton of the constant $\bzero$ process.
    \item A \emph{time-admissible control} $\alpha$ is a process $\Omega \times [0, T] \times [0, T] \ni (\omega, t, t^*) \mapsto \alpha_t^{t^*}(\omega) \in A$ that is $\mathcal{PM}(\mathbb{G}) \otimes \B([0, T])$-measurable such that for almost all $t^*$, $\alpha^{t^*} \in \A(t^*)$. Let $\A$ denote all such strategies.
    \item A \emph{price-admissible control} $\alpha$ is a process $\Omega \times [0, T] \times [P_0, \infty) \ni (\omega, t, p^*) \mapsto \alpha_t(\omega, p^*)\in A$ that is $\mathcal{PM}(\mathbb{G}) \otimes \B([P_0, \infty))$-measurable such that for $\nu_p$-almost all $p^*$, $\alpha_t(\cdot, p^*) = 0 \as$ on the random interval $[0, \T(p^*) \wedge T)$. Let $\A^*$ denote all such strategies.
\end{itemize}
\end{definition}

By Lemma \ref{Lemma:p*measurable}, any time-admissible control $\alpha \in \A$ induces a price-admissible control by $\alpha_t(\cdot, p^*) = \alpha_t^{\T(\cdot, p^*)}(\cdot)$.  For each $\alpha \in \A^*$, the corresponding state process satisfies
\begin{equation}\label{eq:StateProcess}
    X^{\T, \alpha}_t = \ind{t \geq \T}K_0/\mathscr{P} + 
    \intt \alpha_s ds + \sigma (W_{t \vee \T} - W_{\T}), \quad t \in [0, T].
\end{equation}

\subsubsection{Objective and Equilibrium}
Following the same derivation from the N-player game, using \eqref{eq:Nobj} we can define the running cost function $f: [0, T] \times \R \times \R\times [0, T] \times \R \times A \to \R$:
\begin{equation}\label{eq:RunningCost}
f(t, x, \mb, \upeta, \varrho, a) = \kappa(a) + \phi x^2 - x\big(\mb\ind{t < \upeta} + \varrho\ind{t \geq \upeta}\big)
\end{equation}
and the terminal cost function $g: \Omega \times \R \times \R \times [0, T] \to \R$:
\begin{equation}\label{eq:TerminalCost}
g(x_1, x_2, \upeta) = c x_1^2 + \beta_{\upeta}\gamma_{\upeta}x_2.
\end{equation}
Allowing $C >0$ to vary in each step, by Assumptions \eqref{Assumption:B}, \eqref{Assumption:E} and Proposition \ref{Prop:PriceSDEWellPosed} we have
\begin{align*}
    \E\biggl[\sup_{\upeta \in [0, T], \alpha \in \A^*}|g(X^{\T, \alpha}_T, X^{\T, \alpha}_\upeta, \upeta)|^2\biggr] &\leq C\E\sqbra{1 + \pa{\frac{K_0}{p^*}}^4 + \sigma^4W_T^4 + \sup_{\upeta \in [0, T]}\gamma^2_{\upeta} + \sup_{\upeta \in [0, T]}X^2_\upeta}\\
     \leq C\E\sqbra{1 + \pa{\int_0^T b(t, M^P_t, P_t)dt}^2} &\leq C\E\sqbra{1 + \int_0^T |M^{|P|}_t|^2dt}\\
    & \leq C\E\sqbra{1 + T\sup_{t \in [0, T]}P_t^2} < \infty.
\end{align*}
Given a price process $P$, define the process $b^P_t \ce b(t, M^P_t, P_t)$. For a fixed $\theta = (\theta_t)_{t \in [0, T]}$ and $\mu = (\mu_t)_{t\in[0,T]}$, the objective which the representative player minimizes over $\A^*$ is:
\begin{equation*}
J^{\mu, \theta}(\alpha) = \E\sqbra{g(X^{\T, \alpha}_T, X^{\T, \alpha}_{\tau^*(\mu)}, \tau^*(\mu)) + \int_{\T \wedge T}^T f(s, X^{\T,\alpha}_s, b^P_s, \tau^*(\mu), \qv{\rho, \theta_s}, \alpha_s)ds}.
\end{equation*}
If the player does not enter by time $T$, her total cost is $0$. This is also true by construction, see Remark \ref{Remark:C4} below.

\begin{remark}\label{Remark:C}
   We make a note that the following set of properties of the cost structure will be utilized in the proof.
    \begin{enumerate}[label={(C\arabic*)}]
        \item \label{Remark:C1} The running cost function $f: [0, T] \times \R \times \R \times [0, T] \times \R \times A \to \R$ is (jointly) Borel measurable and can be decomposed as 
    \begin{equation*}
        f(t, x, \mb, \upeta, \varrho, a) = f_a(t, x, a) +f_b(t, x, \mb)\ind{0 \leq t < \upeta} + f_c(t,x, \varrho) \ind{t \geq \upeta}.
    \end{equation*}
    For each $t$, $f_a(t, \cdot, \cdot)$, $f_b(t, \cdot, \cdot)$ and $f_c(t, \cdot, \cdot)$ are continuous. In addition, there exists $\ell_f > 0$ such that for all $(t, \mb, \varrho,  x) \in [0, T] \times \R \times \R \times \R$ with $p \leq m$:
    $$|f_a(t, x, a)| + |f_b(t, x, \mb)| + |f_c(t, x,\varrho)| \leq \ell_f\pa{1 + |x|^2 + |\mb|^2}.$$
    \item \label{Remark:C2} 
    The terminal cost function $g: \Omega \times \R \times \R \times [0,T] \to \R$ is almost surely continuous in $(x_1, x_2, \upeta)$. In addition, there exists $C >0$ such that $\E\sqbra{\sup_{\upeta \in [0, T], \alpha \in \A^*}|g(X^{\T, \alpha}_T, X^{\T, \alpha}_\upeta, \upeta)|^2} \leq C.$
    \item \label{Remark:C3} 
    $f$ is strictly convex in $(a, x)$; $g$ is convex in $x_1$ and $x_2$.
    \item \label{Remark:C4}
    $f(t, 0, \mb, \upeta, \varrho, 0) = 0$ for any $(t, \mb, \upeta, \varrho) \in [0, T] \times \R \times [0, T] \times \R$. $g(0, 0, \upeta) = 0$ for all $\upeta \in [0,T]$.
    \end{enumerate}
    Although we will focus on the specific case of the model with cost functions \eqref{eq:RunningCost} and \eqref{eq:TerminalCost}, most of our results remain true for arbitrary costs satisfying \ref{Remark:C1} - \ref{Remark:C4}.
\end{remark}

\subsubsection{Identical Threshold Case}
A special case is where everyone has the same threshold $p^* = P_0$ and thus enters all at the beginning. Then the bubble function $b$ does not depend on $M^p$. Suppose further that $b$ also does not depend on $P$. Then the game reduces to a fixed entry time case in \cite[Proposition A.10]{TangpiWang22}. 

\subsubsection{Common Noise and Admissible Setup}
Unlike idiosyncratic noise, the presence of common noise does not vanish even when the number of players goes to infinity. As a consequence, we need to consider ``random versions'' of $(\mu, \theta)$, which we denote as $(\upmu, \upvartheta)$, to represent the \emph{conditional} probability measures given the common noise. Specifically, the probability setup should also support random variable $\mP \ce (\upmu, \upvartheta): \Omega \to \M \times \Theta^A$. Therefore, for $\alpha \in \A^*$, the objective a representative agent minimizes is
\begin{equation}\label{eq:Objective}
J^{\upmu, \upvartheta}(\alpha) = \E\sqbra{g(X^{\T,\alpha}_T, X^{\T,\alpha}_{\tau^*}, \tau^*(\upmu)) + \int_{\T \wedge T}^T f(s, X^{\T, \alpha}_s, b^{P}_s, \tau^*(\upmu), \qv{\rho, \upvartheta_s}, \alpha_s)ds},
\end{equation}
where $X^{\T, \alpha}$ follows \eqref{eq:StateProcess}.

In our setup, there are two sources of common noise to the players: a Brownian motion $B$ from the price process $P$ and a jump process $D_t = \ind{\tau \leq t}$ for the exogenous burst. For any stochastic process $Z$ and random variable $\xi$, define their natural filtration $\mathbb{F}^{Z, \xi} \ce (\F^{Z, \xi}_t)_{t \in [0, T]}$ where $\F^{Z, \xi}_t$ is the $\P$-completion of $\sigma((\xi, Z_s)_{s \in [0, t]})$. Intuitively, $(\upmu, \upvartheta)$ are conditional laws given $(B, D)$, so if we view $\mP$ as a $\calP(\R) \times \calP(A)$-valued process, it should be $\mathbb{F}^{B, D}$-adapted. The natural filtration to work with is the completion of $\mathbb{F}^{K_0, \mathscr{P}, W, B, D}$. An equilibrium of this type is called a \emph{strong} solution, which is known to be very hard to obtain (see e.g. the monograph  \cite{CarmonaBookII}). Instead, we look for a weak equilibrium in the sense of \cite{Kurtz14, CarmonaCommonNoise} where we only require $(\upmu, \upvartheta)$ to be the conditional law of state and control processes given both the common noise $(B, D)$ and the law process $\mP$ itself. 

We collect all the components from this section in the next definition in a more general setting where we do not assume that the underlying probability space has a product structure.
\begin{definition}
    An \emph{admissible probability setup} is a filtered probability space $(\Omega, \F, \mathbb{G} = (\G_t)_{t \in [0, T]}, \P)$ satisfying the usual conditions that supports the following mutually independent random elements:
\begin{enumerate}
    \item A two-dimensional Brownian motion $(W, B)$.
    \item $\G_0$-measurable initial data $\calI\ce (K_0, \mathscr{P}) \in \R\times [P_0, \infty)$ with law $\nu_K \otimes \nu_p$.
    \item A $\mathbb{G}$-stopping time $\tau$ with law $\nu_\tau$, from which we can define the jump process $D_t \ce \ind{\tau \leq t}$.
\end{enumerate}
\end{definition}
If an admissible probability setup additionally supports $\mP = (\upmu, \upvartheta)$ taking values in $\M \times \Theta^A$, we can then define $\tau^*(\upmu)$, the price process $P$ and random entry time
$$\T \ce \inf\{t \in [0, T]: P_t - \mathscr{P} \geq 0\} \wedge (T + 1).$$
Observe that $\T$ may not be defined for \emph{every} threshold value $p^* \in [P_0, \infty)$ that $\mathscr{P}$ takes, making this setup slightly weaker. Similarly, we will also weaken the notion of price-admissibility and let $\A^*$ denote the set of processes $\Omega \times [0, T] \ni (\omega, t) \mapsto \alpha_t(\omega) \in A$ that is $\mathbb{G}$-progressive measurable such that $\P$-almost surely, $\alpha_t\ind{t \in [0, \T)} = 0$. In fact, Lemma \ref{Lemma:p*measurable} ensures that our $p^*$-by-$p^*$ construction is also a particular case under this new definition.

It is worth noting that if $\tau$ is also independent from $\mP$, then $\tau$ will be an $\mathbb{F}^{\calI, B, W, D, \mP}$-totally inaccessible stopping time (see Remark \ref{Remark:Inaccessible}). In particular, if $\mathbb{G}$ is just $\mathbb{F}^{\calI, B, W, D, \mP}$, this would be a desired feature for the exogenous burst time because the admissible controls can only react to it once $\tau$ occurs but cannot anticipate it.

\begin{definition}\label{Def:weakMFGsolution}
A weak MFG equilibrium with strong control is an admissible probability setup $(\Omega, \F, \mathbb{G}, \P)$ that supports a $\G_T$-measurable random variable $\mP = (\hat\upmu, \hat\upvartheta):\Omega \to \M \times \Theta^A$, paired with optimal control $\hat{\alpha} \in \A^*$ and corresponding state process $X^{\T, \hat{\alpha}}$ satisfying \eqref{eq:StateProcess} such that
\begin{enumerate}
\item The filtration $\mathbb{G} = \mathbb{F}^{\calI, B, W, D, \mP}$.
\item $\hat\alpha$ minimizes over $\A^*$ the objective $J^{\hat{\upmu}, \hat{\upvartheta}}$ defined in \eqref{eq:Objective}. 
\item $\hat\upmu$ is a version of the conditional law of $X^{\T,\hat{\alpha}}$ given $(B, D, \mP)$ under $\P$. That is,
$$\hat\upmu_t(\cdot) = \P \pa{X_t^{\T, \hat{\alpha}} \in \cdot | \F_t^{B, D, \mP}} \text{ for almost all } t \in [0, T].$$
\item $\hat\upvartheta$ is a version of the conditional law of $\hat\alpha$ given $(B, D, \mP)$ under $\P$. That is,
$$\hat\upvartheta_t(\cdot) = \P \pa{\hat\alpha_t \in \cdot | \F_t^{B, D, \mP}} \text{ for almost all } t \in [0, T].$$
\end{enumerate}
\end{definition}
\begin{theorem}\label{Theorem:MFGExistence}
    Under Assumptions \eqref{Assumption:B} and \eqref{Assumption:E}, there exists a weak MFG equilibrium with strong control. 
\end{theorem}
\section{Existence of MFG Solutions with Weak Control}\label{Section:WeakControl}
\subsection{Weak Controls}
The term ``strong control'' in Definition \ref{Def:weakMFGsolution} refers to the fact that $\hat{\alpha}$ is an $A$ valued process that is $\mathbb{F}^{\calI, B, W, D, \mP}$-progressive. We shall prove Theorem \ref{Theorem:MFGExistence} by following the chain of arguments presented in \cite{CarmonaCommonNoise}. Specifically, we use a fixed point and compactness argument by discretizing the common noise $(B, D)$ and then taking weak limit to obtain an equilibrium. To ensure that the limit exists, we first work with \emph{relaxed controls} in a larger filtration.

\subsubsection{Relaxed Controls}
Since the space of uniformly bounded functions is not compact, a standard workaround when analyzing extended MFGs, especially in the presence of common noise, is to consider \emph{relaxed} controls. A relaxed control is a randomized strategy taking values in $\Gamma$ where 
\begin{equation*}
    \Gamma \ce \{\upgamma \in \calP([0, T] \times A) \text{ with time marginal being the Lebesgue measure}\}.
\end{equation*}
Any $\upgamma \in \Gamma$ can be characterized, with dt a.s. uniqueness, by the form $\upgamma(dt, da) = (\upgamma_t(da)dt)_{t \in [0, T]}$ where $t \mapsto \upgamma_t \in \calP(A)$ is a Borel measurable mapping. Therefore, we can view each $\upgamma \in \Gamma$ as a $\calP(\calP(A))$-valued process. For a given admissible probability setup, the set of admissible relaxed controls is defined as 
$$\bbGamma\ce \cbra{\upgamma \in \Gamma \text{ that is } \mathbb{G}\text{-progressive such that $\P$-almost surely } \upgamma_t\ind{t \in [0, \T)} = \delta_{0} \ \forall t \in [0,T]}.$$
A \emph{strict} control refers to the case where $\upgamma_t$ is $\P$-almost surely a Dirac measure almost everywhere. The state process corresponding to a relaxed control $\upgamma \in \bbGamma$ is 
\begin{equation}\label{eq:StateProcess_Relaxed}
    X^{\T, \upgamma}_t = \ind{t \geq \T}K_0/\mathscr{P} + 
    \intt\int_A a\upgamma(ds, da) + \sigma (W_{t \vee \T} - W_{\T}), \quad t \in [0, T].
\end{equation}
Define $\Theta$ as the subset of $\mathcal{P}([0, T] \times \mathcal{P}(A))$ whose first projection is Lebesgue measure $dt$ on $[0, T]$. Any $\theta \in \Theta$ can be characterized, with $dt$ a.s. uniqueness, by $\{\theta_t \in \mathcal{P}(\mathcal{P}(A))\}_{t \in [0, T]}$ such that $\theta(dt, dq) = \theta_t(dq)dt$. We naturally extend any bounded measurable function $F: \mathcal{P}(A)\to \R$ to $\underline{F}: \mathcal{P}(\mathcal{P}(A)) \to \R$ by 
\begin{equation*}
\underline{F}(\theta) \ce \int_{\mathcal{P}(A)}\theta(dq)F(q).   
\end{equation*}
In particular, $\underline{F}(\delta_q) = F(q)$ for $q \in \mathcal{P}(A)$. Recall from Remark \ref{Remark:C1} that we have separability between $a$ and $q$ in the cost $f$. Therefore, when evaluating $f$ (or rather its extension) on an element of $\calP(\calP(A))$, we can drop the underline from the notation to avoid further confusion. In particular, for a bounded measurable function $\rho: A \to \R$, sometimes we slightly abuse the notation by using $\qv{\rho, \theta_t}$ to mean $\qv{\rho, \int_{\calP(A)}\theta_t(dq)}$ if $\theta$ is in $\Theta$ instead of $\Theta^A$. Endow $\Theta$ with the stable topology, which is the weakest topology making the map $\theta \to \int\phi d\theta$ continuous, for each bounded measurable function $\phi: [0, T] \times \mathcal{P}(A) \to \R$ that is continuous in the measure variable for each $t$. Since $A$ is convex, compact and metrizable, so is $\Theta$. See \cite{JacodMemin81} for details. 

The version of objective function \ref{eq:Objective} for relaxed controls is 
\begin{equation}\label{eq:Objective_Relaxed}
J^{\upmu, \upvartheta}(\upgamma) = \E\sqbra{g(X^{\T,\upgamma}_T, X^{\T,\upgamma}_{\tau^*}, \tau^*(\upmu)) + \int_{\T \wedge T}^T \int_A f(s, X^{\T,\upgamma}_s, b^{P}_s, \tau^*(\upmu), \qv{\rho, \upvartheta_s}, a)\upgamma(da,ds)}.
\end{equation}
Notice from \eqref{eq:Objective_Relaxed} and \eqref{eq:StateProcess_Relaxed} that $A$-valued controls are naturally embedded in the space of relaxed controls in the form of strict controls. 

\subsubsection{Immersion Property and Lifted Environment}
In this section, we will also weaken the first requirement in Definition \ref{Def:weakMFGsolution} and work with a filtration $\mathbb{G}$ that is potentially larger than $\mathbb{F}^{\calI, B, W, D, \mP}$. Allowing more information into the system immediately requires extra care to ensure fairness in observing that additional information. A widely-used procedure is to check that $\mathbb{F}^{\calI, B, W, D, \mP}$ is \emph{immersed} in $\mathbb{G}$. This notion of fairness is also called the (H)-hypothesis, natural extension, or compatibility. It is a crucial property in the theory of filtration enlargement \cite{KharroubiLim14, PE2}, stochastic control \cite{Kurtz14, KarouiNguyenJeanblanc87}, the theory of conditional McKean-Vlasov SDEs \cite{Lacker2020SuperpositionAM} and of course mean field games \cite{CarmonaCommonNoise}.

\begin{definition}\label{Definition:Immersion}
    A filtration $\mathbb{H}$ is said to be \emph{immersed} in another filtration $\mathbb{F}$ defined on the same probability space if $\mathbb{H} \subset \mathbb{F}$ and every square integrable $\mathbb F$-martingale is a square integrable $\mathbb H$-martingale. An $\mathbb{F}$-adapted c\`{a}d-l\`{a}g process $\upmu = (\upmu_t)_{t \geq 0}$ with values in a Polish space is \emph{compatible} with $\mathbb{F}$ if its natural filtration $\mathbb{F}^\upmu$ is immersed in $\mathbb{F}$.
\end{definition}

The following proposition is a useful characterization of this property and explains how compatibility weakens the adaptedness to a conditional independence requirement, which is mainly a property of laws. See e.g. \cite[Proposition 1.3]{CarmonaBookII} for a proof.

\begin{proposition}\label{Prop:Immersion}
    On probability space $(\Omega, \F, \P)$, consider two filtrations $\mathbb{H} = (\H_t)_{t \in [0, T]} \subset \mathbb{F} = (\F_t)_{t \in [0, T]}$. The following statements are equivalent.
    \begin{enumerate}
        \item $\mathbb{H}$ is immersed in $\mathbb{F}$.
        \item $\H_T$ is conditionally independent from $\F_t$ given $\H_t$ for every $t \in [0, T]$.
        \item For any $\zeta \in \mathbb{L}^1(\F_t)$, $\E[\zeta | \H_T] = \E[\zeta |\H_t]$.
    \end{enumerate}
\end{proposition}
Specifically, the third statement in Proposition \ref{Prop:Immersion} allows us to eventually recover a strong control in Section \ref{Section:StrongControl} from a larger filtration. To ensure that we carry enough information in the smaller filtration for the immersion property to eventually hold, we will in a \emph{lifted environment} \cite{CarmonaBookII}. Instead of $\mP = (\upmu, \upvartheta)$, we require the admissible probability setup to support a random \emph{joint} probability measure $\mM$ that represents the conditional law of $(X^\T, W, \upgamma, \calI)$ given common noise. Let $\mM^x$ and $\mM^\upgamma$ denote its first and third marginals, which serves the same purpose of $(\upmu, \upvartheta)$ in the objective \eqref{eq:Objective_Relaxed}. Now we have all the ingredients to define a solution with weak control.

\begin{definition}\label{Def:weakMFGsolution_Relaxed}
A weak MFG equilibrium with weak control is an admissible probability setup $(\Omega, \F, \mathbb{G}, \P)$ that supports a $\G_T$-measurable random variable $\mM: \Omega \to \calP(\X^*\times \X\times \Gamma\times \R^2)$, paired with optimal \emph{relaxed} control $\hat{\upgamma} \in \bbGamma$ and corresponding state process $X^{\T, \hat{\upgamma}}$ satisfying \eqref{eq:StateProcess_Relaxed} such that
\begin{enumerate}
\item The filtration $\mathbb{F}^{\calI, B, W, D, \mM}$ is immersed in $\mathbb{G}$.
\item $\hat\upgamma$ minimizes over $\bbGamma$ the relaxed objective $J^{\mM^x, \mM^{\upgamma}}$ defined in \eqref{eq:Objective_Relaxed}.
\item $\mM^x$ is a version of the conditional law of $X^{\T,\hat{\upgamma}}$ given $(B, D, \mM)$ under $\P$. 
\item $\mM^{\upgamma}$ is a version of the conditional law of $\hat{\upgamma}$ given $(B, D, \mM)$ under $\P$.
\end{enumerate}
\end{definition}

It is worth noting that both definitions are weak in the probabilistic sense, where the probability space is part of the solution. They are also both weak in the sense of control theory, where the equilibrium strategy is not necessarily measurable with respect to the Brownian motions, but potentially depends on additional randomness. 

The usual fixed point argument using compactness no longer applies to these conditional probability measures as their domain becomes too large. To combat the infinite dimensionality issue, \citet{CarmonaCommonNoise} discretizes time and space to reduce the common noise to a finite dimension process and then pass to the limit. We adapt the discretization scheme from \cite{CarmonaBookII}, also used in \cite{Campi21}. In this section, our goal is to prove the following intermediate result.
\begin{theorem}\label{Theorem:Existence_Weak}
    Under Assumptions \eqref{Assumption:B} and \eqref{Assumption:E}, there exists a weak MFG equilibrium with weak control.
\end{theorem}

\subsection{Weak Formulation and Enlargement of Filtration}
Since the probability space is part of the solution, it is convenient to work on the canonical space with the product structure in Section \ref{Subsub:EntryTime}. We will also work under the weak formulation as in \cite{CarmonaLacker15}. Define 
$$\Omega_1 \ce \R \times \X, \quad \Omega_0 = \X \times \R_+, \quad \Omega \ce \Omega_1 \times \Omega_0, \quad \Omega_c \ce \Omega \times [P_0, \infty)$$
and let $(K_0, W, B, \tau, \mathscr{P})$ be the corresponding identity maps. Let $\F$ be a $\sigma$-algebra carrying the above random variables. Define the corresponding probability measures $$\Q_1 \ce \nu_K\otimes \W, \quad \Q_0 \ce \W \otimes \nu_\tau , \quad \Q \ce \Q_1 \otimes \Q_0,\quad \P \ce \Q \otimes \nu_p.$$
Define entry time $\T$ in a $p^*$-by-$p^*$ way on $\Omega_c$ as \eqref{eq:EntryTime}. Lemma \ref{Lemma:p*measurable} ensures that we have an admissible probability setup. Let $X^\T$ denote the uncontrolled state variable on the product space $[0,T]\times \Omega_c$:
    \begin{equation}\label{eq:DriftlessState}
        X^{\T}_t \ce K_0/\mathscr{P} + \sigma(W_t - W_{\T}) \text{ for } t \geq \T \quad \text{and} \quad X^{\T}_t \ce 0 \text{ for } t \in [0, \T).
    \end{equation}
Given $\alpha \in \A^*$ define
\begin{equation}\label{eq:Girsanov}
    \frac{d\P^{\alpha}}{d\P} \ce \mathcal{E}\Big(\int_0^\cdot \sigma^{-1}\alpha_sd W_s \Big)_T, \quad W^{\alpha}_t \ce W_t - \int_0^t \sigma^{-1}\alpha_sds.
 \end{equation} 
 By Girsanov's theorem, and square integrability of $\alpha$, $W^\alpha$ is a Brownian motion under $\P^\alpha$ and $X^\T$ follows the state SDE \eqref{eq:StateProcess} under $\P^\alpha$. Given $(\upmu, \upvartheta)$, the cost under the weak formulation is
\begin{equation}\label{eq:MFGobjectiveWeak}
    J^{\upmu, \upvartheta}_{weak}(\alpha) \ce \E^{\P^{\alpha}}\sqbra{g(X^{\T}_T,X^{\T}_{\tau^*(\upmu)},\tau^*(\upmu)) + \int_{\T\wedge T}^Tf(s, X^{\T}_s, b^P_s, \tau^*(\upmu), \qv{\rho, \upvartheta_s},\alpha_s)ds}.\\
\end{equation}
If we fix a price threshold $p^*$, then $\Q^{\alpha(p^*)}$ is defined on $\Omega$ by $\frac{d\Q^{\alpha(p^*)}}{d\Q}$ in a similar way as $\P^{\alpha}$.

\subsubsection{Progressive Enlargement of Filtration}
We now recall some facts regarding filtration enlargement. Let $\mathbb{F} = (\F_t)_{t \in [0, T]}$ be a filtration supporting $(W, B, \calI)$ and is independent from $\tau$. Let $(\mathcal D_t)_{t \in [0,T]}$ denote the natural filtration of the exogenous burst time jump process $(D_t)_{t\in [0, T]}$. Define the progressively enlarged filtration 
\begin{equation*}
    \G_t \ce \F_t \vee \mathcal D_t, \qquad \mathbb G = (\G_t)_{t \in [0, T]}.
\end{equation*}
Note that $\mathbb{G}$ is the smallest filtration which contains $\mathbb{F}$ and such that $\tau$ is a $\mathbb{G}$-stopping time. Since $\tau$ is independent from $\mathbb{F}$, Proposition \ref{Prop:Immersion} implies that $\mathbb{F}$ is immersed in $\mathbb{G}$. In particular, $W$ and $B$ remain ($\mathbb G, \P$)--Wiener processes. 

For any $\mathbb{G}$--predictable process $h$, Assumption \ref{Assumption:E2} implies there exists a unique $\mathbb{F}$-predictable process $f$ such that $h_t\ind{t \leq \tau} = f_t\ind{t \leq \tau}$ (see \cite[Page 186 (a)]{bookDellacherie92}). Since $D$ is a $\mathbb{G}$--submartingale, by Doob--Meyer decomposition we can find a unique, $\mathbb{F}$--predictable, increasing compensator process $K$ with $K_0= 0$ and such that:
\begin{equation}\label{eq:compensatorMG}
M^{\tau}_t = D_t - \intt (1- D_{s-})dK_s
\end{equation}
is a $(\P, \mathbb{G})$--martingale. Under Assumption \ref{Assumption:E2}, we have $dK_t = k_tdt$.
\begin{remark}\label{Remark:Inaccessible}
The random time $\tau$ is a $\mathbb{G}$-inaccessible stopping time if either of the two following conditions is satisfied (see e.g. \cite{Azema93, BlanchetJeanblanc04}):
\begin{enumerate}
\item Every $\mathbb{F}$-martingale is a.s. continuous;
\item $\tau$ avoids all $\mathbb{F}$-stopping times. That is, $\P(\tau = L) = 0$ for any $\mathbb{F}$-stopping time $L$. 
\end{enumerate}
For example, if $\mathbb{F}$ is just the $\P$-completed Brownian filtration with the initial enlargement of $K$ and $\mathscr{P}$, by martingale representation theorem (1) would be satisfied. Under the non-atomic condition in Assumption \ref{Assumption:E2}, if there is independence between $\tau$ and $\mathbb{F}$, then (2) holds.  
\end{remark}

\subsection{Proof of Theorem \ref{Theorem:Existence_Weak}}\label{Subsection:BSDE_RandomEntry}
In this section, we prove the existence of equilibrium with weak controls using backward stochastic differential equations (BSDE). 
\subsubsection{Backward SDEs with Random Entry Times}
We begin by introducing a few notation of spaces and norms. For a filtration $\mathbb H$ and probability measure $\widetilde{\Q}$ on $\Omega$, define the following spaces of processes on $[s, t] \subseteq [0, T]$:
\begin{itemize}
    \item Let $\mathcal{S}^2_{\mathbb H, \widetilde{\Q}}[s,t]$ denote the space of $\R$-valued $\mathbb H$-progressively measurable, c\`{a}dl\`{a}g processes $Y$ on $\Omega\times [s, t]$ satisfying $$||Y||_{\mathcal{S}_{\mathbb{H},\widetilde{\Q}}^2[s, t]} \ce \E^{\widetilde{\Q}}\sqbra{\sup_{u \in [s, t]}\abs{Y_u}^2}^{\frac{1}{2}} < \infty.$$
    \item Let $\mathcal{H}^2_{\mathbb{H}, \widetilde{\Q}}[s,t]$ denote $\R$-valued $\mathbb{H}$-predictable processes $Z$ on $\Omega\times [s, t]$ satisfying $$||Z||_{\mathcal{H}_{\mathbb{H},\widetilde{\Q}}^2[s, t]} \ce \E^{\widetilde{\Q}}\sqbra{\int_s^t |Z_u|^2du}^{\frac{1}{2}} < \infty.$$ 
    \item Let $\mathcal{H}_{\mathbb{H}, \widetilde{\Q}, D}^2[s,t]$ denote $\R$-valued $\mathbb{H}$-predictable processes $U$ on $\Omega\times [s, t]$ satisfying $$||U||_{\mathcal{H}^2_{\mathbb{H},\widetilde{\Q}, D}[s, t]} \ce \E^{\widetilde{\Q}}\sqbra{\int_s^t |U_u|^2dD_u}^{\frac{1}{2}} < \infty.$$ 
\end{itemize}
We drop $\widetilde{\Q}$ from notation when $\Q$ is the probability measure. Respectively, for a probability measure $\widetilde{\Q}_c$ on $\Omega_c$, define $\mathcal{S}^2_{\mathbb H, c, \widetilde{\Q}_c}[s,t], \mathcal{H}^2_{\mathbb{H}, c, \widetilde{\Q}_c}[s,t]$ and $\mathcal{H}_{\mathbb{H}, c, \widetilde{\Q}_c, D}^2[s,t]$ in the same way for processes on $\Omega_c \times [s, t]$. In particular, when $\widetilde{\Q}_c = \P$, $\mathcal{H}^2_{\mathbb{H}, c}[s,t]$ denotes $\R$-valued $\mathbb{H}$-predictable processes $Z$ on $\Omega_c \times [s, t]$ satisfying $$||Z||_{\mathcal{H}_{\mathbb{H}, c}^2[s, t]} \ce \E^{\P}\sqbra{\int_s^t |Z_u|^2du}^{\frac{1}{2}} = \E^{\Q}\sqbra{\int_{[P_0, \infty)]}\int_s^t |Z_u(p^*)|^2du\nu_p(dp^*)}^{\frac{1}{2}} < \infty.$$ 
We drop $[s, t]$ from notation when considering the whole interval $[0, T]$. 

Since we take the weak formulation to MFGs, we can rewrite the objective function \eqref{eq:MFGobjectiveWeak} using the solution to a BSDE. Define the Hamiltonian by
    \begin{equation*}
    H(t, x, \mb, \upeta, \varrho, z, a) = f(t, x, \mb, \upeta, \varrho, a) + \sigma^{-1}az.
    \end{equation*}
By Remark \ref{Remark:C} and Assumption \ref{Assumption:E1}, for each $(t, x, m, p, \upeta, \varrho, z)$, there exists a unique element in $A$ that minimizes $H(t, x, m, \upeta, \varrho, z, \cdot)$. For our model, the minimizer is a function of $z$ only, which we denote as $\hat{a}(z)$. Let $h$ denote the minimized Hamiltonian, that is 
\begin{equation}\label{eq:MinHamiltonian}
h(t, x, \mb,\upeta, \varrho, z) \ce H(t, x, \mb, \upeta, \varrho, z, \hat{a}(z)) = \kappa(\hat{a}(z)) + \phi x^2 - x\big(\mb\ind{t < \upeta} + \varrho\ind{t \geq \upeta}\big) + \sigma^{-1}\hat{a}(z)z.
\end{equation}
\begin{remark}\label{Remark:S}
    We point out some properties of $\hat{a}$ and $h$ that will be utilized later. 
    \begin{enumerate}[label={(S\arabic*)}]
        \item \label{Remark:S1} For a general $f$ and $g$ satisfying the properties in Remark \ref{Remark:C}, $\hat{a}$ is a jointly measurable function of $(t, x, z)$ and continuous in $z$ by Berge's maximum theorem. In our case, the unique minimizer $\hat{a}(\cdot)$ only depends on $z$. 
        \item \label{Remark:S2} The minimized Hamiltonian $h$ is Lipschitz in $z$, and it is jointly continuous in $(x, z, \varrho)$ for fixed $(t, m, p, \upeta)$.
    \end{enumerate}
\end{remark}

Recall the definition of $M^\tau$ in (\ref{eq:compensatorMG}). For a given $p^*$, consider a generic type of BSDEs on the enlarged filtration $\mathbb{G}$ solved on $[\T(p^*), T]$
\begin{equation}\label{eq:BSDEp*}
    \begin{split}
    Y_t & = g(X^{\T(p^*)}_T, X^{\T(p^*)}_{\tau^*}, \tau^*(\upmu)) + \int_t^T h(s,X^{\T(p^*)}_s, b^P_s, \tau^*(\upmu), \qv{\rho, \upvartheta_s}, Z_s)ds\\
    & - \int_t^T Z_sdW_s - \int_t^T \mathfrak{Z}_s dB_s- \int_t^T U_sdM^\tau_s - \int_t^T dM_s, \quad t \in [\T(p^*), T].
    \end{split}
\end{equation}
where $M$ is a martingale orthogonal to $(W, B, M^\tau)$. A solution to the BSDE \eqref{eq:BSDEp*} is a process $(Y, Z, \mathfrak{Z}, U, M) \in \mathcal{S}^2_{\mathbb G}[\T(p^*), T] \times \mathcal{G}^2_{\mathbb{G}}[\T(p^*), T] \times \mathcal{H}^2_{\mathbb{G}} [\T(p^*), T]\times \mathcal{H}^2_{\mathbb{G}, D}[\T(p^*), T] \times \mathcal{S}^2_{\mathbb G}[\T(p^*), T]$ on the probability space $(\Omega, \mathcal{F}, \mathbb{G}, \P)$. If the pre-enlarged filtration $\mathbb{F}$ is generated by the Brownian motions, then $M \equiv 0$. Note that the BSDE above is solved on a random interval even after conditioning on a $p^*$. The following proposition addresses the solvability of this BSDE. To differentiate the two types of admissibility, we denote a time-admissible control in $\A$ by $\alpha$ and price-admissible control in $\A^*$ by $\alpha^{\T}$.

\begin{proposition}\label{Prop:OptimalControl}
    Suppose that $\mathbb{G} = \mathbb{F}^{\calI, W, B, D}$ and fix a $\mathbb{G}$-progressive $\mathfrak{P} = (\upmu, \upvartheta)$. Given $p^*$, for each $t^* \in [0, T]$, there exists a unique solution $(Y^{t^*}, Z^{t^*}, \mathfrak{Z}^{t^*}, U^{t^*})$ to the following BSDE 
    \begin{equation}\label{eq:BSDE_Brownian}
    \begin{split}
    Y_t & = g(X^{t^*}_T, X^{t^*}_{\tau^*}, \tau^*(\upmu)) + \int_t^T h(s,X^{t^*}_s, b^P_s, \tau^*(\upmu), \qv{\rho, \upvartheta_s}, Z_s)ds\\
    & - \int_t^T Z_sdW_s - \int_t^T \mathfrak{Z}_s dB_s- \int_t^T U_sdM^\tau_s, \quad t \in [t^*, T].
    \end{split}
    \end{equation}
    where $X^{t^*}$ follows $$X^{t^*}_t = \ind{t \geq t^*}\pa{K_0/p^* + \sigma(W_t - W_{t^*})}.$$
    If we define the process $\hat\alpha^{t^*}_t = \ind{t \geq t^*}\hat{a}(Z^{t^*}_t) \in \A^{t^*}$ for each $t^*$, then $\hat\alpha$ is time admissible and induces a price-admissible control $\hat\alpha^{\T(\cdot)} \in \A^*$. Moreover, $\hat\alpha^{\T}$ minimizes \eqref{eq:MFGobjectiveWeak} over $\A^*$.
\end{proposition}
\begin{proof}
    Using Assumption \ref{Assumption:E2} and \eqref{eq:compensatorMG}, we can rewrite \eqref{eq:BSDE_Brownian} as
    \begin{equation*}
    \begin{split}
    Y_t & = g(X^{t^*}_T, X^{t^*}_{\tau^*}, \tau^*(\upmu)) + \int_t^T h(s,X^{t^*}_s, b^P_s, \tau^*(\upmu), \qv{\rho, \upvartheta_s}, Z_s) + U_sk_s\ind{0 \leq s < \tau}ds\\
    & - \int_t^T Z_sdW_s - \int_t^T \mathfrak{Z}_s dB_s- \int_t^T U_sdD_s, \quad t \in [t^*, T].
    \end{split}
    \end{equation*}
    
    Well-posedness follows from \cite[Theorem 53.1]{KarouiMazliakBSDE}. We need to show that $\hat{\alpha}$ is jointly measurable when composing the $t^*$-by-$t^*$ solutions. We first show that $t^* \mapsto \hat{\alpha}^{t^*}$ is $\P$-almost surely left-continuous in $\mathcal{H}^2_{\mathbb{G}}$. 

    Suppose we have a sequence $t_n^* \uparrow t^* \in [0, T]$, and let $\alpha^{t^*_n}$ and $\alpha^{t^*}$ be the corresponding control processes. Then we have
    \begin{align*}
        \norm{\alpha^{t^*_n} - \alpha^{t^*}}_{\mathcal{H}_{\mathbb{G}}^2} & = \E\sqbra{\intT|\alpha^{t^*_n}_t - \alpha^{t^*}_t|^2dt} = \E\sqbra{\int_{t^*_n}^{t^*}|\alpha^{t^*_n}_t|^2dt} + \E\sqbra{\int_{t^*}^{T}|\alpha^{t^*_n}_t - \alpha^{t^*}_t|^2dt}\\
        & = \E\sqbra{\int_{t^*_n}^{t^*}|\alpha^{t^*_n}_t|^2dt} + \E\sqbra{\int_{t^*}^{T}\abs{\hat{a}(Z^{t^*_n}_t) - \hat{a}(Z^{t^*}_t)}^2dt}.
    \end{align*}
    The first term goes to $0$ by dominated convergence theorem since $A$ is assumed to be bounded. To show the convergence of the second term, by continuity of $\hat{a}$ it suffices to show $Z^{t^*_n} \nto Z^{t^*}$ in $\mathcal{H}^2_{\mathbb{G}}[t^*, T]$. By the stability of BSDE solutions (e.g. \cite[Proposition 54.2]{KarouiMazliakBSDE}), we have
    \begin{align*}
        &\norm{Z^{t^*_n} - Z^{t^*}}^2_{\mathcal{H}_\mathbb{G}^2[t^*, T]} = \E\sqbra{\int_{t^*}^{T}\abs{Z^{t^*_n}_t -  Z^{t^*}_t}^2dt} \leq C\E\sqbra{\abs{g(X^{t^*_n}_T, X^{t^*_n}_{\tau^*}, \tau^*) - g(X^{t^*}_T, X^{t^*}_{\tau^*}, \tau^*)}^2}\\
        & \qquad + C\E\sqbra{\int_{t^*}^T\abs{h(s,X^{t^*_n}_s, b^P_s, \tau^*, \qv{\rho, \upvartheta_s}, Z^{t^*}_s) - h(s,X^{t^*}_s, b^P_s, \tau^*, \qv{\rho, \upvartheta_s}, Z^{t^*}_s)}^2ds}.
    \end{align*}
    It is easy to check that for any $t$, we have $X^{t^*_n}_t \nto X_t^{t^*}$, and by \ref{Remark:S2} left-continuity is proved. The second condition in Lemma \ref{Lemma:Caratheodory} is satisfied, and invoking the lemma yields time-admissibility.

    Therefore, for each $p^* \in [P_0, \infty)$, we can solve \eqref{eq:BSDEp*} with random entry time on $[\T(p^*), T]$ and obtain $\hat{\alpha}^{\T(p^*)}$. Note that altering $p^*$ only affects the initial inventory $K_0/p^*$ and entry time $\T(p^*)$. Since $P_0 > 0$, and also by Lemma \ref{Lemma:p*measurable}, the function $p^* \mapsto \hat\alpha^{\T(p^*)}_t(\omega)$ is left-continuous for $\Q\times\lambda$-almost every $(\omega, t)$. Therefore, Lemma \ref{Lemma:Caratheodory} again implies joint measurability and therefore price-admissibility of $\hat{\alpha}^{\T(\cdot)}$.

    Observe that \eqref{eq:MFGobjectiveWeak} can be rewritten as $J^{\upmu, \upvartheta}_{weak}(\alpha^\T) = \E^{\nu_p}[J^{\upmu, \upvartheta}_{weak}(\alpha^\T | p^*)]$ for each price-admissible $\alpha^\T$, where $J^{\upmu, \upvartheta}_{weak}(\cdot | p^*)$ is the conditional objective given $\mathscr{P} = p^*$. That is
    \begin{equation*}
        J^{\upmu, \upvartheta}_{weak}(\alpha^\T | p^*) \ce \E^{\Q^{\hat{\alpha}^{\T(p^*)}}}\sqbra{g(X^{\T(p^*)}_T, X^{\T(p^*)}_{\tau^*}, \tau^*(\upmu)) + \int_{\T(p^*) \wedge T}^Tf(s,X^{\T(p^*)}_s, b^P_s, \tau^*(\upmu), \qv{\rho, \upvartheta_s}, \alpha^{\T(p^*)}_s)ds}.
    \end{equation*}
    Then optimality of $\hat{\alpha}^{\T}$ follows if we show conditional optimality of $\hat{\alpha}^{\T(p^*)}$ for each $p^*$, which we fix from this point on. Take any price-admissible strategy $\beta^{\T}$. We can uniquely solve the following BSDE
    \begin{equation}\label{label:ComparisonBSDE}
    \begin{split}
    Y^{\beta, p^*}_t & = g(X^{\T(p^*)}_T, X^{\T(p^*)}_{\tau^*}, \tau^*(\upmu)) + \int_t^T H(s,X^{\T(p^*)}_s, b^P_s, \tau^*(\upmu), \qv{\rho, \upvartheta_s}, Z^{\beta, p^*}_s, \beta^{\T(p^*)}_s)ds\\
    & - \int_t^T Z^{\beta, p^*}_sdW_s - \int_t^T \mathfrak{Z}^{\beta, p^*}_s dB_s- \int_t^T U^{\beta, p^*}_sdM^\tau_s , \quad t \in [0, T].
    \end{split}
\end{equation}
    We can also solve \eqref{label:ComparisonBSDE} on $[0, T]$ with $\hat{\alpha}$ as input. Then by \eqref{eq:MinHamiltonian} and uniqueness, the solution coincides with the solution of \eqref{eq:BSDEp*} on the interval $[\T(p^*), T]$. Comparison principle of \eqref{label:ComparisonBSDE}(\cite[Proposition 4.3]{TangpiWang22}) implies $Y^{\hat{\alpha}^{\T}, p^*}_0 \leq  Y^{\beta^\T, p^*}_0$ $\Q$-almost surely. Optional stopping theorem and Remark \ref{Remark:C4} imply
    \begin{equation*}
        J^{\upmu, \upvartheta}_{weak}(\hat{\alpha}^{\T}|p^*) = \E^{\Q^{\hat{\alpha}^{\T(p^*)}}}\sqbra{Y^{\hat\alpha, p^*}_0} = \E^{\Q}\sqbra{Y^{\hat\alpha, p^*}_0} \leq \E^{\Q}\sqbra{Y^{\beta, p^*}_0} =\E^{\Q^{\beta^{\T(p^*)}}}\sqbra{Y^{\beta, p^*}_0} = J^{\upmu, \upvartheta}_{weak}(\beta^\T|p^*).
    \end{equation*} 
\end{proof}
For  the remainder of this section up until Remark \ref{Remark:Candidate}, we take $\mathbb{G} = \mathbb{F}^{\calI, W, B, D}$, so Proposition \ref{Prop:OptimalControl} applies. This result implies that for a given $(\upmu, \upvartheta)$ we can find $\hat{\alpha}$ by solving the problem $p^*$-by-$p^*$, and it is well defined for every $p^*$, not just $\nu_p$ almost every $p^*$. On the other hand, by uniqueness of the optimizer, we can also obtain $\hat{\alpha}^\T$ by solving the BSDE on the whole product space.

\begin{corollary}\label{Cor:BSDEProductSpace}
Given $(\upmu, \upvartheta): \Omega_0 \to (\M, \Theta)$ that is $\mathbb{F}^{B, D, \calI}$-progressive, there exists a unique solution $(Y, Z, \mathfrak{Z}, U) \in \mathcal{S}^2_{\mathbb{G}, c} \times \mathcal{G}^2_{\mathbb{G}, c} \times \mathcal{H}^2_{\mathbb{G},c} \times \mathcal{H}^2_{\mathbb{G},c,D}$ to the following BSDE on the product space in the $\mathbb{G}$ filtration
    \begin{equation}\label{eq:BSDE_Main}
    \begin{split}
    Y_t & = g(X^{\T}_T, X^{\T}_{\tau^*(\upmu)}, \tau^*(\upmu)) + \int_{t}^T h(s,X^{\T}_s, b^P_s, \tau^*(\upmu), \qv{\rho, \upvartheta_s}, Z_s)ds\\
    & - \int_{t}^T Z_sdW_s - \int_{t}^T \mathfrak{Z}_s dB_s- \int_{t}^T U_sdM^\tau_s, \quad t \in [\T \wedge T, T].
    \end{split}
    \end{equation}
Additionally, the process defined by $\hat\alpha_t \ce \ind{t \geq \T}\hat{a}(Z_t) \in \A^*$ is $\P \otimes dt$ almost surely identical to the one constructed in Proposition \ref{Prop:OptimalControl} and minimizes \eqref{eq:MFGobjectiveWeak}.
\end{corollary}

For any price-admissible $\alpha \in \A^*$, let $\widetilde{X}^{\alpha}$ denote the solution of (\ref{eq:StateProcess}) defined on $\Omega_c \times [0, T]$. Then by Girsanov's theorem 
\begin{equation}\label{eq:GirsanovLaw}
\P \circ (\widetilde{X}^{\alpha}, \calI, \T, W, B)^{-1} = \P^{\alpha} \circ (X^\T, \calI, \T, W^\alpha, B)^{-1}.
\end{equation}
Under Assumption \ref{Assumption:E1}, for any $p > 0$ we have
    \begin{equation}\label{eq:Xestimate}
    \sup_{\alpha \in \A^*}\E^{\P^{\alpha}}\sqbra{\sup_{p^* \in [P_0, \infty)}\norm{X^{\T(p^*)}}_{\X^*}^p} = \sup_{\alpha \in \A^*}\E\sqbra{\sup_{p^* \in [P_0, \infty)}\norm{\widetilde{X}^{\alpha^{\T(p^*)}}}_{\X^*}^p} < \infty.
    \end{equation}

\subsubsection{Fixed Point from Discretization}
Instead of conditional measure flows given common noise $(B, D)$, we look at a piecewise constant approximation process. Suppose for $N \in \N$, a partition $\{0 = t_0 < t_1 <\cdots < t_{N-1} < t_N = T\}$ on $[0, T]$ and a finite set (some grid on space) $\Lambda_N \subset \R$ are given. Define the $\Lambda_N$-valued finite process $V^N$ on $\X$ by 
\begin{equation}\label{eq:V_N}
    V^N_t(\beta) \ce \sum_{i=1}^N v_{i-1}(\beta)\ind{t \in [t_{i-1}, t_i)} + v_{N-1}(\beta) \ind{t = T},
\end{equation}
where each $v_i: \X \to \Lambda_N$ is a $\F^B_{t_i}$ measurable random variable taking values in the finite set, for $i = 0, \dots, N-1$. These knots are meant to approximate the Brownian common noise $B_{t_i}(\beta)$. We also have an additional source of common noise: the external burst time $\tau$, which requires us to discretize the jump process $D$. 

Let $\X_D$ denote the space of processes on $[0, T]$ of the form $D_t(\upeta) = \ind{\upeta \leq t}$ for some $\upeta \in [0, T]$. Equip $\X_D$ with the natural metric $d(D(\upeta), D(\upeta')) = |\upeta - \upeta'|$. For $N \in \N$, define the $\X_D$ valued process on $\R_+$ by:
\begin{equation}\label{eq:D_N}
    D^N_t(\upeta) \ce \sum_{i=1}^N \ind{\upeta \leq t_{i-1}}\ind{t \in [t_{i-1}, t_i)} +  \ind{\upeta \leq t_{N-1}}\ind{t = T}.
\end{equation}
Then it is obvious that for all $\varepsilon > 0$:
\begin{equation}\label{eq:D^NProperty}
    \lim_{N \to \infty}\P^N\pa{d(D^N, D) \leq \varepsilon}  = 1.
\end{equation} 

Let $\mathcal{V}_N \ce \{A_1, \cdots, A_{|\mathcal{V}_N|}\}$ denote the (finite) $\sigma$-algebra generated by $(V^N, D^N)$, and choose $v_i$'s such that $\P(A_k) > 0$ for every $k$. We now define the input domain for conditional laws of the state and control. At this stage, we can work with $A$-valued controls. However, in anticipation of taking the limit in the space of relaxed controls, we switch to strict controls now. For $\alpha \in \A^*$, call $\upgamma(\alpha) \in \bbGamma$ its corresponding strict control. Define
\begin{align}
    \label{Definition:Mn_discretization}
    \M_N & \ce \cbra{\pa{\L^{\alpha, 1}(X^\T), \dots, \L^{\alpha,|\mathcal{V}_N|}(X^\T)}: \text{ for some }\alpha \in \A^*} \subset \calP_1(\X^*)^{|\mathcal{V}_N|}\\
    \label{Definition:Thetan_discretization}
    \Theta_N & \ce \cbra{\pa{\L^{\alpha, 1}(\upgamma(\alpha)), \dots, \L^{\alpha,|\mathcal{V}_N|}(\upgamma(\alpha))}: \text{ for some }\alpha \in \A^*} \subset \Theta^{|\mathcal{V}_N|}
\end{align}
where for each $\alpha \in \A^*, k\in 1, 2, \dots, |\mathcal{V}_N|$, $\L^{\alpha, k}$ denotes the conditional law under $\P^\alpha$ given $A_k$. Equip each coordinate of $\M_N$ with the Wasserstein metric. Denote by $\vmm = (\mm_1, \dots, \mm_{|\mathcal{V}_N|})$ an element in $\M_N$. 
\begin{lemma}\label{Lemma:mu_bar_cont}
 The process $t \mapsto \E^{\P^\alpha}[X^\T_t | A_k]$ is $\Q_0$-almost surely continuous for each $k$ for all $\alpha \in \A^*$. Consequentially, the mapping $\M_N \ni \vmm \mapsto \bar\tau\pa{\upmu^N}$ is $\Q_0$-almost surely continuous on the closure of $\M_N$, where for $\omega_0 = (\beta, \upeta) \in \Omega_0 = \X \times \R_+$, $\upmu^N((\beta, \upeta)) \ce \sum_{k=1}^{|\mathcal{V}_N|} \mm_{k}\ind{(\beta, D(\upeta)) \in A_k}$.
\end{lemma}

\begin{proof}
     Take a sequence $t_n \nto t_\infty$ in $[0, T]$. Note that the event that $X^\T_{t_n}$ does not converge to $X^\T_{t}$ is, up-to a $\P^\alpha$-null set, contained in $\{\T = t_\infty\}$. Then by dominated convergence theorem, it suffices to show that for all $t$, $\P^\alpha(\{\T = t_\infty\}|A_k) = 0$ for each $A_k$, which is implied by $\P^{\alpha}(\T = t_\infty) = \P(\T = t_\infty) = 0$. This follows from price dynamics \eqref{eq:PriceDynamics+} and the fact that the only jump of $P$ is negative by Assumptions \eqref{Assumption:B}. Then continuity of mean processes for each $k$ implies continuity of the endogenous burst mapping (see the proof of \cite[Theorem 6.1]{TangpiWang22}). Taking closure in Wasserstein space preserves the continuity of the mean processes.
\end{proof}



\begin{lemma}\label{Lemma:Mcompact}
    The set $\M_N$ is Wasserstein pre-compact and convex in $\calP(\X^*)^{|\mathcal{V}_N|}$, and $\Theta_N$ is Wasserstein compact and convex in $\Theta^{|\mathcal{V}_N|}$.
\end{lemma}
\begin{proof}
    The statement on $\Theta_N$ is immediate given that $A$ is compact and convex. For convexity of $\M_N$, take $\alpha^1, \alpha^2 \in \A^*$. It suffices to show that for all $\lambda \in [0, 1]$, there is $\alpha \in \A^*$ such that
\begin{equation*}
    \frac{d\P^{\alpha}}{d\P} = \lambda \frac{d\P^{\alpha^1}}{d\P} + (1-\lambda)\frac{d\P^{\alpha^2}}{d\P}.
\end{equation*}
See \cite[Lemma 4.7]{TangpiWang22} for a proof. As for Wasserstein pre-compactness, we show that for each marginal. Since the dimensionality is finite, it suffices to check for each $k \in \{1, \dots, |\mathcal{V}_N|\}$ which follows from \cite[Lemma 5.9]{TangpiWang22} with the additional fact that $\E[\xi|A_k] \leq \E[\xi]/\P(A_k)$ for arbitrary non-negative random variable $\xi$, and $\P(A_k) > 0$ for all $k$.
\end{proof}
Suppose we take any $\vmm^N = (\mm^N_1, \dots, \mm^N_{|\mathcal{V}_N|}) \in \M_N$ and $\vma^N = (\ma^N_1, \dots, \ma^N_{|\mathcal{V}_N|}) \in \Theta_N$. Define the inputs to the optimization problem: for each $(\beta, \upeta) = \omega_0 \in \Omega_0 = \X \times \R_+$,
\begin{equation}\label{eq:conditional_measures_N}
\upmu^N(\omega_0) \ce \sum_{k=1}^{|\mathcal{V}_N|} \mm^N_{k}\ind{(\beta, D(\upeta)) \in A_k} \in \calP_1(\X), \quad \upvartheta^N(\omega_0) \ce \sum_{k=1}^{|\mathcal{V}_N|}\ma^N_k\ind{(\beta, D(\upeta)) \in A_k} \in \Theta.
\end{equation}
 Again $(\upmu^N, \upvartheta^N)$ can be viewed as a $\mathbb{F}^{B, D}$-measurable random process taking values $(\upmu_t^N, \upvartheta_t^N)$ in $\calP(\R) \times \calP(\calP(A))$. By Proposition \ref{Prop:OptimalControl} we obtain strict, optimal control $\hat{\alpha}^{\T, N}\in \A^*$ with $\hat{\alpha}^{\T, N}_t = \hat{a}(Z^{\T, N}_t)$ along with the probability measure $\P^{N} \ce \P^{\hat{\alpha}^{\T, N}}$ on $\Omega_c$, and $W^N_t \ce W_t - \intt \sigma^{-1}\hat{\alpha}^{\T, N}_sds$ is a Brownian motion under $\P^N$. Denote by $\upgamma^N$ the optimal control in the strict relaxed form $\upgamma(\hat{\alpha}^{\T, N}) = (\delta_{\hat{\alpha}_t^{\T, N}}dt)_{t \in [0, T]}$. Define output conditional measures $(\vmm^{N,out},\vma^{N, out}) \in \M_N \times \Theta_N$:
\begin{equation}\label{eq:output_mu_N}
    \mm^{N,out}_k(\cdot) \ce \frac{\P^N(A_k \cap \{X^{\T}\in \cdot\})}{\P^N(A_k)}, \quad \ma^{N, out}_k(\cdot) \ce \frac{\P^N(A_k \cap \{\upgamma^N \in \cdot\})}{\P^N(A_k)}.
\end{equation}
We have now defined our fixed point mapping: $$\Phi^N: \M_N \times \Theta_N \ni (\vec{\mathfrak{m}}^{N}, \vec{\mathfrak{a}}^{N}) 
 \mapsto (\vec{\mathfrak{m}}^{N,out}, \vec{\mathfrak{a}}^{N,out}) \in \M_N \times \Theta_N.$$
\begin{lemma}\label{Lemma:ContFixedPoint_N}
    For each $N \in \N$, the mapping $\Phi^N$ is continuous.  
\end{lemma}
\begin{proof}
    We shall fix $N$ and drop the notation to avoid confusion with the proof steps for sequential continuity. Recall $\hat{\alpha}^{\T}$ is obtained by solving the optimization problem $p^*$-by-$p^*$. Let $\P^{\hat{\alpha}^{\T(p^*)}}$ denote the conditional probability measure on $\Omega$ from Girsanov transform for a given $p^* \in [P_0, \infty)$. Take a sequence of vectors $(\vmm^n, \vma^n) \nto (\vmm^\infty, \vma^\infty) \in \M_N \times \Theta_N$ in the sense that each coordinate converges in Wasserstein distance, and define correspondingly $(\upmu^n, \upvartheta^n)$ for $n \in \N \cup \{\infty\}$. Then $\Q_0$ almost surely, $\upmu^n$ converges to $\upmu^\infty$ in $\M$ and $\upvartheta^n$ converges to $\upvartheta^\infty$ in $\Theta$. For $n \in \N \cup \{\infty\}$ define the discretized conditional probability measures $\upmu^n$ and $\upvartheta^n$ as in \eqref{eq:conditional_measures_N}, and let $Z^{\T(p^*), n}$ be part of the unique solution to the BSDE \eqref{eq:BSDEp*} with $M\equiv 0$ and input $(\upmu^n, \upvartheta^n)$. Let $\hat{\alpha}^{\T(p^*), n}_t = \hat{a}(Z^{\T(p^*), n}_t)$ be the optimal control given $\mathscr{P} = p^*$.  
    

    We first show that $\hat{\alpha}^{\T(p^*), n} \nto \hat{\alpha}^{\T(p^*), \infty}$ in $\mathcal{H}^2_{\mathbb{G}}$ for \emph{every} $p^* \in [P_0, \infty)$. Recall that optimal controls are continuous in $Z^{\T(p^*)}$, so it suffices to show $Z^{\T(p^*), n} \nto Z^{\T(p^*), \infty}$ in $\mathcal{H}^2_{\mathbb{G}}$. Using the stability property of BSDE solutions \cite[Proposition 54.2]{KarouiMazliakBSDE}, convergence of $Z^{\T(p^*), n}$ is immediately implied if we show $$\limn \E\sqbra{\abs{\Delta_ng}^2+ \int_{\T(p^*)}^T\abs{\Delta_n h_s}^2ds} = 0,$$
    where $\Delta_n g \ce g\pa{X^{\T(p^*)}_T, X^{\T(p^*)}_{\tau^*(\upmu^n)}, \tau^*(\upmu^n)} - g\pa{X^{\T(p^*)}_T, X^{\T(p^*)}_{\tau^*(\upmu^\infty)}, \tau^*(\upmu^\infty)}$ and 
    $$\Delta_n h_s \ce h(s,X^{\T(p^*)}_s, b^P_s, \tau^*(\upmu^n), \qv{\rho, \upvartheta^n_s}, Z^{\T(p^*), \infty}_s) - h(s,X^{\T(p^*)}_s, b^P_s, \tau^*(\upmu^\infty), \qv{\rho, \upvartheta^\infty_s}, Z^{\T(p^*), \infty}_s).$$
By continuity of $g$ we know that $\Delta_n g$ converges to $0$ in probability if both $\tau^*(\upmu^n) \nto \tau^*(\upmu^\infty)$ and $X_{\tau^*(\upmu^n)} \nto X_{\tau^*(\upmu^\infty)}$ in probability. Lemma \ref{Lemma:mu_bar_cont} gives us convergence in burst time. Also, the inventory is continuous everywhere except at entry $\T(p^*)$. Therefore, $X_{\tau^*(\upmu^n)}$ does not converge in probability to $ X_{\tau^*(\upmu^\infty)}$ only if $\T(p^*) = \tau^*(\upmu^\infty)$, which also has probability $0$. Therefore,
    \begin{equation*}
   \P(\{|\Delta_n g|^2 \nto 0\})  \geq 1 - \P\pa{\{\T(p^*) = \tau^*(\upmu^\infty)\}} = 1.
   \end{equation*}
  By dominated convergence theorem, we have $\E[\abs{\Delta_ng}^2] \nto 0$.

Now let $I_n$ denote the random interval $[\tau^*(\upmu^n) \wedge \tau^*(\upmu^\infty),\tau^*(\upmu^n) \vee \tau^*(\upmu^\infty)]$. Observe that $$\abs{\ind{0 \leq t < \tau^*(\upmu^n)} - \ind{0 \leq t < \tau^*(\upmu^\infty)}}^2 = \abs{1-\ind{t \geq \tau^*(\upmu^n)} - 1+ \ind{t \geq \tau^*(\upmu^\infty)}}^2= \ind{t \in I_n}.$$
Define $I^{p^*}_n \ce I_n \cap [\T(p^*), T]$. Remark \ref{Remark:C1} implies
    \begin{align*}
        \E\sqbra{\int_{\T(p^*)}^T\abs{\Delta_n h_s}^2 ds} & \leq \E\sqbra{\int_{\T(p^*)}^T\abs{f_b(s, X^{\T(p^*)}_s, b^P_s)\ind{0 \leq t < \tau^*(\upmu^n)} - f_b(s, X^{\T(p^*)}_s, b^P_s)\ind{0 \leq t < \tau^*(\upmu^\infty)}}^2ds}\\
        & + \E\sqbra{\int_{\T(p^*)}^T\abs{f_c(s, X^{\T(p^*)}_s, \qv{\rho, \upvartheta^n_s})\ind{s \geq \tau^*(\upmu^n)} - f_c(s, X^{\T(p^*)}_s, \qv{\rho, \upvartheta^\infty_s})\ind{s \geq \tau^*(\upmu^\infty)}}^2ds}\\
        & \leq \E\sqbra{\int_{I_n^{p^*}}|f_b(s, X^{\T(p^*)}_s, b^P_s)|^2ds} + 2 \E\sqbra{\int_{I_n^{p^*}}|f_c(s, X^{\T(p^*)}_s, \qv{\rho, \upvartheta^\infty_s})|^2ds} \\
        & + 2\E\sqbra{\int_{\T(p^*)}^T\abs{f_c(s, X^{\T(p^*)}_s, \qv{\rho, \upvartheta^n_s})- f_c(s, X^{\T(p^*)}_s, \qv{\rho, \upvartheta^\infty_s})}^2\ind{s \geq \tau^*(\upmu^n)}ds}.
    \end{align*}

    Lemma \ref{Lemma:mu_bar_cont} implies $\ind{t \in I_n^{p^*}} \nto \ind{t = \tau^*(\upmu^\infty)}$ almost surely. Under Assumption \ref{Assumption:E1} and Remark \ref{Remark:C1}, all terms converge to $0$ by dominated convergence theorem and Fubini's theorem. Therefore, for all $p^* \in [P_0, \infty)$, $\hat{\alpha}^{\T(p^*), n} \nto \hat{\alpha}^{\T(p^*), \infty}$ in $\mathcal{H}^2_{\mathbb{G}}$. Since $A$ is assumed to be bounded, this implies that $\hat{\alpha}^n \nto \hat{\alpha}^{\infty}$ in the $\mathcal{H}^2_{\mathbb{G}, c}$ sense on the product space $\Omega_c$.

    For each $n \in \N \cup \{\infty\}$, denote by $\P^n \ce \P^{\hat{\alpha}^{n}}$ the probability measure on $\Omega_c$ from Girsanov transformation. By construction, $\P^{n} \ll \P$ for all $n$ and $$\frac{d\P^{n}}{d\P^{\infty}} = \calE\pa{\int_0^\cdot \sigma^{-1}\pa{\hat{\alpha}^{n}_t - \hat{\alpha}^{\infty}_t}dW_t}_T.$$ 
    Therefore, by boundedness of $A$ we can calculate the relative entropy
    \begin{equation*}
        \mathcal{H}(\P^{\infty}|\P^{n}) = - \E^{\P^\infty}\sqbra{\log \frac{d\P^{n}}{d\P^{\infty}}} = \frac{1}{2}\E^{\P^\infty}\sqbra{\int_{0}^T\sigma^{-2}\abs{\hat{\alpha}^{n}_t - \hat{\alpha}^{\infty}_t}^2dt}\nto 0,
    \end{equation*}
    Pinsker's inequality implies that $\P^{n}$ converges to $\P^{\infty, p^*}$ in total variation. By triangular inequality and the convergence in $\mathcal{H}^2_{\mathbb{G}, c}$ of controls, we have $$\P^{n} \circ \pa{\hat{\alpha}^{n}_t}^{-1} \nto \P^{\infty} \circ \pa{\hat{\alpha}_t^{\infty}}^{-1} \text{ in } dt\text{-measure}.$$
    Bounded convergence theorem yields $\P^n \circ (\upgamma(\hat{\alpha}^n))^{-1} \nto \P^\infty \circ (\upgamma(\hat{\alpha}^\infty))^{-1}$ in the stable topology. Since $\P(A_k) > 0$ for each $k$, this implies convergence of $\ma^{n, out}_{k} \nto \ma^{\infty, out}_{k}$ as well. Boundedness of $A$ ensures Wasserstein convergence as well. $\P^{n} \nto \P^{\infty}$ in total variation also implies weak convergence of $\mm^{n, out}_{k} \nto \mm^{\infty, out}_{k}$ for each $k$. To show Wasserstein convergence, it suffices to show uniform integrability (see e.g. \cite[Theorem 6.9]{villani09}).
    \begin{align*}
        & \lim_{R \to \infty}\limsup_{n\to\infty}\sup_{k \in \{1, \dots, |\calV_N|\}}\int_{\{\norm{\bx}_{\X^*} > R\}}\norm{\bx}_{\X^*} \mm^{n, out}_k(d\bx)\\
        & \leq \lim_{R \to \infty}\sup_{\alpha \in \A^*}\sup_{k \in \{1, \dots, |\calV_N|\}}\frac{1}{\P(A_k)}\E^{\P^{\alpha}}\sqbra{\norm{X^{\T}}_{\X^*}\ind{\norm{X^\T}_{\X^*} > R}} = 0
    \end{align*}
    which follows from \eqref{eq:Xestimate}. Therefore, since $|\calV_N|$ is finite: $$\Phi^N(\vmm^n, \vma^n) \nto \Phi^N(\vmm^\infty, \vma^\infty) \text{ in } \M_N \times \Theta_N,$$ 
    and continuity of $\Phi^N$ holds for all $N$. 
\end{proof}

\begin{proposition}\label{Prop:FixedPoint}
    The mapping $\Phi^N$ admits a fixed point $(\vmm^N, \vma^N) \in \M_N \times \Theta_N$ for all $N \in \N$.
\end{proposition}
\begin{proof}
Let $\overline{\M}_N$ denote the closure of $\M_N$, which by Lemma \ref{Lemma:Mcompact} is convex and compact. Note that the input $(\upmu^N, \upvartheta^N)$ to the BSDE is still well-defined for $\vmm^N \in \overline{\M}_N$, except that $\upvartheta^N$ might not be the law of a strict control anymore. Therefore, we can define $\Phi^N$ on the larger domain $\overline{\M}_N \times \Theta_N$. Moreover, the Wasserstein closure preserves the continuity of the mean process. This implies that $\tau^*(\cdot)$ is still continuous on $\overline{\M}_N$ (see the proof of \cite[Theorem 6.1]{TangpiWang22}), so continuity of $\Phi^N$ still holds. Applying Brouwer's fixed point theorem (e.g. \cite[Corollary 17.56]{bookInfDimAnalysis06}) yields a fixed point $(\vmm^N, \vma^N)$ of $\Phi^N$. However, the range of $\Phi^N$ is still strictly in $\M_N \times \Theta_N$ since allowing $\upmu$ to take values in $\overline{\M}_N$ doesn't affect how we construct the optimal control and its corresponding state process. Then $(\vmm^N, \vma^N) \in \M_N \times \Theta_N$ is also a fixed point of $\Phi$ had we not enlarged the domain.
\end{proof}
For each $N$, let $\alpha^N \in \A^*$ be the equilibrium strategy from a fixed point $(\vmm^N, \vma^N)$ from Proposition \ref{Prop:FixedPoint} and $\P^N = \P^{\alpha^N}$. Define the corresponding random laws $(\upmu^N, \upvartheta^N)$ as in \eqref{eq:conditional_measures_N}. Then by construction, $\alpha^N$ minimizes $J^{\upmu^N, \upvartheta^N}_{weak}$ defined in \eqref{eq:MFGobjectiveWeak} over $\A^*$.

\begin{remark}\label{Remark:Optimality}
    For each $N \in \N$, we work with the same filtration $\mathbb{F}^{\calI, W, B, D}$. In fact, by the argument in \cite[Proposition A.10]{TangpiWang22}, $\hat{\alpha}^N$ also minimizes \eqref{eq:MFGobjectiveWeak} had we allowed for a larger filtration $\widetilde{\mathbb{F}} \supseteq \mathbb{F}^{\calI, W, B, D}$. We provide a proof to keep the paper self-contained.
\end{remark}
\begin{proof}
    We fix $N \in \N$ and recall that the uncontrolled state process is defined by $$X^{\T}_t \ce K_0/\mathscr{P} + \sigma(W_t - W_{\T}) \text{ for } t \geq \T \quad \text{and} \quad X^{\T}_t \ce 0 \text{ for } t \in [0, \T).$$
Recall from Remark \ref{Remark:S} that $\hat{a}(\cdot)$ is a continuous function that minimizes the Hamiltonian. For each $\beta \in \A^*$ defined on the probability space $(\Omega, \F, \P, \widetilde{\mathbb{F}})$, by \cite[Theorem 53.1]{KarouiMazliakBSDE} there exists a unique solution $(Y^{\beta}, Z^{\beta}, \mZ^{\beta}, U^{\beta}, N^{\beta}) \in \mathcal{S}^2_{\widetilde{\mathbb{F}}} \times \mathcal{G}^2_{\widetilde{\mathbb{F}}} \times \mathcal{H}^2_{\widetilde{\mathbb{F}}}\times \mathcal{H}^2_{\widetilde{\mathbb{F}}, D} \times \mathcal{S}^2_{\widetilde{\mathbb{F}}}$ to the following BSDE
    \begin{equation}\label{label:BSDEgeneral_intermediate}
        \begin{split}
    Y_t & = g(X^{\T}_T, X^{\T}_{\tau^*(\upmu^N)}, \tau^*(\upmu^N)) + \int_t^T H(s,X^{\T}_s, b^P_s, \tau^*(\upmu^N), \qv{\rho, \upvartheta^N_s}, Z_s, \beta_s)ds\\
    & - \int_t^T Z_sdW_s - \int_t^T \mathfrak{Z}_s dB_s- \int_t^T U_sdM^\tau_s - \int_t^T dN_s, \quad t \in [0, T].
    \end{split}
    \end{equation}
    Remark \ref{Remark:S2} also implies well-posedness for the following BSDE solved on $(\Omega, \F, \P, \widetilde{\mathbb{F}})$:
    \begin{equation}\label{eq:BSDEgeneral}
        \begin{split}
    Y_t & = g(X^{\T}_T, X^{\T}_{\tau^*(\upmu^N)}, \tau^*(\upmu^N)) + \int_t^T \ind{s \geq \T}h(s,X^{\T}_s, b^P_s, \tau^*(\upmu^N), \qv{\rho, \upvartheta^N_s}, Z_s)ds\\
    & - \int_t^T Z_sdW_s - \int_t^T \mathfrak{Z}_s dB_s- \int_t^T U_sdM^\tau_s - \int_t^T dN_s, \quad t \in [0, T].
    \end{split}
    \end{equation}
    whose unique solution we denote by $(\hat{Y}, \hat{Z}, \hat{\mZ}, \hat{U}, \hat{N})$. However, since $(\mu^N, \upvartheta^N)$ are $\mathbb{F}^{B, D}$ measurable, uniqueness of the solution implies that $\hat{N}$ is $\P^N \otimes dt$ almost surely zero, and $(\hat{Y}, \hat{Z}, \hat{\mZ}, \hat{U})$ coincides with the solution of \eqref{eq:BSDE_Main} on $[\T, T]$. Recall from Corollary \ref{Cor:BSDEProductSpace} and the construction of the fixed point mapping $\Phi^N$ that $\P^N \otimes dt$ almost surely, we must also have $\alpha^N_t = \hat{a}(\hat{Z}_t)\ind{t \geq \T}$.
    
    Recall from Remark \ref{Remark:C4} that for any admissible control, the Hamiltonian is $0$ before entry. Then the generator of \eqref{label:BSDEgeneral_intermediate} is $\P \otimes dt$-almost surely greater than the generator of \eqref{eq:BSDEgeneral}, and they are equal when we take $\beta = \alpha^N$. If a comparison principle for the general BSDE \eqref{label:BSDEgeneral_intermediate} holds, then Remark \ref{Remark:Optimality} follows from the argument in the proof of Proposition \ref{Prop:OptimalControl}. To ease some notation, for $t \in [0, T]$ we denote 
    \begin{equation*}
        \Delta H_t(\beta, \alpha^N) \ce H(t, X^{\T}_t, b^P_t, \tau^*(\upmu^N), \qv{\rho, \upvartheta^N_t}, Z^{\beta}_t, \beta_t) - H(t, X^{\T}_t, b^P_t, \tau^*(\upmu^N), \qv{\rho, \upvartheta^N_t}, Z^{\alpha^N}_t, \alpha^N_t).
    \end{equation*}
    
    Due to the presence of compensated martingale $M^\tau$ and orthogonal martingale $N$, additional conditions are required for comparison principle to hold. In light of \cite{Cohen10}, a sufficient condition is the existence of an equivalent measure $\widetilde{\P}$ to $\P$ such that 
    \begin{equation}\label{label:ComparisonSufficientCondition}
    \begin{split}
        S_t & \ce -\inttT \Delta H_s(\beta, \alpha^N)ds + \inttT (Z^{\beta}_s - Z^{\alpha^N}_s)dW_s + \inttT (\mZ^{\beta}_s - \mZ^{\alpha^N}_s)dB_s\\
        &\qquad + \inttT(U^{\beta}_s - U^{\alpha^N}_s)dM^{\tau}_s + \inttT dN^{\beta}_s - \inttT dN^{\alpha^N}_s, \quad t\in[0,T]
    \end{split}
    \end{equation}
    is a martingale under $\widetilde{\P}$. This probability can be constructed by defining $$\frac{d\widetilde{\P}}{d\P} \ce \calE\pa{\int_0^{\cdot}\frac{\Delta H_t(\beta, \alpha^N)}{Z^{\beta}_t - Z^{\alpha^N}_t}\ind{Z^{\beta}_t - Z^{\alpha^N}_t \neq 0}dW_t}_T.$$
    Since $H$ is Lipschitz in $z$, we can apply Girsanov's theorem and conclude by \cite[Theorem 2]{Cohen10} that comparison principle holds. Denote again by $\P^\beta$ the probability measure corresponding to $\beta$, which is defined as in \eqref{eq:Girsanov}. Recall that $\P^N, \P^\beta$ and $\P$ all agree at time $t = 0$ and the Hamiltonian is $0$ for any admissible strategy before $\T$. Using comparison principle and applying optional stopping theorem yield $$J^{\upmu^N, \upvartheta^N}_{weak}(\alpha^N) = \E^{\P^N}\sqbra{Y^{\alpha^N}_{\T}} = \E^{\P^N}\sqbra{Y^{\alpha^N}_0} = \E^{\P}\sqbra{Y^{\alpha^N}_0} \leq \E^{\P^\beta}\sqbra{Y^{\beta}_0}=\E^{\P^\beta}\sqbra{Y^{\beta}_\T} = J^{\upmu^N, \upvartheta^N}_{weak}(\beta).$$
    The inequality holds for all $\mathbb{F}$-admissible $\beta$.
\end{proof}
\begin{remark}\label{Remark:Relaxed_vs_Strict}
    Since the state variable is linearly controlled, for each $\upgamma \in \bbGamma$ we can define an $A$-valued control $\alpha_t = \int_A a \upgamma_t(da)$ in $\A^*$ that corresponds to the same state process. Moreover, Remark \ref{Remark:C3} and Jensen's inequality imply that this control $\alpha$ is at least as good as $\upgamma$ in terms of minimizing the objective value. See \cite[Theorem A.9]{HaussmannLepeltier90} and \cite[Theorem 4.11]{KarouiNguyenJeanblanc87} for more rigorous arguments on existence of strict controls achieving the same value and regarding the weak formulation of the problem. In other words, given inputs $(\upmu, \upvartheta)$, the optimal $A$-valued control is also optimal over relaxed controls for the relaxed objective \eqref{eq:Objective_Relaxed}.
\end{remark}


\subsubsection{Solutions as Weak Limit Points}
Before taking $N \to \infty$, we shall ``lift'' the marginal laws of $\X^\T$ and $\upgamma$ to a joint law in order to carry more information along the way. Let $\upgamma^N = \upgamma(\alpha^N)$ be the strict control corresponding to the equilibrium strategy $\alpha^N$ from the fixed point mapping.  Recall that $\calI = (K_0, \mathscr{P})$, taking values in $\R \times [p^*, \infty)$ denotes the initial wealth and entry threshold. Define on $\Omega_0$ the lifted random law $\mM^N  \in \calP(\X^* \times \X \times \Gamma \times \R^2)$ by
\begin{equation}\label{eq:LiftedEnvironment}
\mM^N(\omega_0) \ce \sum_{k = 1}^{|\calV_N|}\ind{(\beta, \upeta) \in A_k}\L^{\alpha^N, k}\pa{(X^\T, W^N, \upgamma^N, \calI)}, \quad \omega_0 = (\beta, \upeta) \in \Omega_0.
\end{equation}
Recall the notation from Definition \ref{Def:weakMFGsolution_Relaxed} that $\mM^{N, x}$ and $\mM^{N, \upgamma}$ denote the first and third marginals of $\mM^N$, which are the conditional law of $(X^{\T}, \upgamma^N)$ under $(V^N, D^N)$.
\begin{lemma}\label{Lemma:TightnessLifted}
The sequence $\P^N \circ (B, \tau, \mM^N, X^\T, W^N, \upgamma^N, \calI)^{-1}$ is tight.
\end{lemma}

\begin{proof}
    Assumption \ref{Assumption:E1} implies that 
    $$\E^{\P^N}\sqbra{\sup_{t\in[0, T]} |X^\T_t|} \leq C$$
    for some $C > 0$ that does not depend on $N$. Let $s$ be a $\mathbb{G}$ stopping time and $\delta > 0$. Then 
    \begin{equation*}
    \E^{\P^N}\sqbra{|X^\T_{s + \delta} - X^\T_s|} \leq \E^{\P^N}\sqbra{\int_{s}^{s + \delta}|\hat{\alpha}^N_t|dt + \sigma|W^N_{\rho + \delta} - W^N_{\rho}|} \leq C\delta^{1/2}
    \end{equation*}
    for a possibly different $C$. These two conditions are sufficient for Aldous' criterion for tightness of $\P^N \circ (X^\T)^{-1}$. For each $N$, $(W^N, B)$ remain independent Brownian motions under $\P^N$, and $(\tau, \calI)$ also retain the same law throughout by independence. Compactness of $A$ implies the tightness of $\P^N \circ (\upgamma^N)^{-1}$. 
    
    Now tightness of $\P^N\circ (\mM^N)^{-1}$ in the weak topology follows from the tightness of $\P^N \circ (X^\T, W^N, \upgamma^N, \calI)^{-1}$ (see the proof of \cite[Lemma 3.16]{CarmonaBookII}). As suggested by Lemma \ref{Lemma:mu_bar_cont}, we will need to equip $\calP(\X^* \times \X \times \Gamma \times \R^2)$ with the Wasserstein metric in order to guarantee continuity of $\bar\tau$, where the product space is equipped with the $l_1$ metric. By \cite[Corollary 5.6]{CarmonaBookI} and boundedness of $A$, the proof of \cite[Lemma 3.16]{CarmonaBookII} implies that it suffices to show uniform square-integrability of $\norm{X^\T}_{\X^*} + \norm{W^N}_\infty$. Since $W^N$ is a Brownian motion under $\P^N$, we only need to show that $$\lim_{R \to \infty}\sup_{N\in\N}\E^{\P^N}\sqbra{\norm{X^\T}_{\X^*}\ind{\norm{X^\T}_{\X^*} \geq R}} < \infty,$$
    which is implied by \eqref{eq:Xestimate}. Finally, tightness of the marginals implies that of the joint law. 
    
    \end{proof}

\begin{remark}\label{Remark:Candidate}
    Let $\bar\P \ce \P^\infty \circ (B, \tau,\mM^\infty, X^\T, W^\infty, \upgamma^\infty, \calI)^{-1}$ be a limit point in Lemma \ref{Lemma:TightnessLifted}. Since we work with a weak limit, we only care about the law, not the processes themselves. Therefore, without loss of generality, we can take $\bar\Omega_0 \ce \R_+ \times \X \times \calP_1(\X^*\times \X \times \Gamma \times \R^2), \ \bar\Omega_1 \ce \X^* \times \X \times \Gamma \times \R^2$. Then let $(\tau, B, \mM^\infty)$, $(X, W^\infty, \upgamma^\infty, \calI)$ be the respective canonical processes on $\bar\Omega \ce \bar\Omega_0 \times \bar\Omega_1$, and $\bar\P$ a probability measure on $\bar\Omega$.  For each $N$, define also on $\bar{\Omega}$ the law $\barP^N \ce \P^N \circ (B, \tau,\mM^N, X^\T, W^N, \upgamma^N, \calI)^{-1}$. We can obviously drop the $\infty$ from the notation (or even $N$, since we can always work on the canonical space), but we keep it to avoid confusion.
\end{remark}

Define the jump process $D$, price process $P$ and entry time $\T$, thanks to the strong solvability from Proposition \ref{Prop:PriceSDEWellPosed}. Take $\mathbb{G}$ to be the $\bar\P$-completed natural filtration generated by $(B, D,\mM^\infty, X^\T, W^\infty, \upgamma^\infty, \calI)$, which is again the progressively enlarged version of $\mathbb{F}^{B, \mM^\infty, X, W^\infty, \upgamma^\infty, \calI}$ by $D_t = \ind{\tau \leq t}$, so $\tau$ is a $\mathbb{G}$-inaccessible stopping time. Let $\F = \G_T \vee \sigma(\tau)$. Weak convergence of $\upgamma^N$ to $\upgamma^\infty$ implies that $X^\T$ with entry time $\T$ satisfies the relaxed state SDE \eqref{eq:StateProcess_Relaxed} on $(\bar\Omega, \F, \bar\P, \mathbb{G})$.

Note that the existence of fixed points of the discretized problem and tightness result both hold for arbitrary choice of discretized process $V^N$. Now we need to specify the time and space grids to ensure $V^N$ approximates $B$ well enough. We will prove Theorem \ref{Theorem:Existence_Weak} by verifying that $\mM^\infty$ defined on $(\bar\Omega, \F, \bar\P, \mathbb{G})$ satisfies the three required conditions in Definition \ref{Def:weakMFGsolution_Relaxed}, namely consistency, compatibility and optimality. 

\subsubsection{Consistency}
Under a slight abuse of notation of re-indexing $N$, we take the same processes used in \cite{CarmonaBookII, Campi21} where time is discretized to the dyadic mesh and space projected to a more refined grid. Specifically, for a fixed $N \in \N$, let $t_i = \frac{iT}{2^N}$ for $i = 0, \dots, 2^N$. Set $v_0 = 0$ and $v_i = v_{i-1} + \Pi^{(N)}(B_{t_i} - B_{t_{i-1}})$, where the projection map $\Pi^{(N)}:\R \to \R$ is defined as  
\begin{equation*}
    \Pi^{(N)}(x) = 
    \begin{cases}
        4^{-N}\lfloor 4^Nx\rfloor & |x| \leq 4^N\\
        4^N \sign(x) & |x| > 4^N
    \end{cases}.
\end{equation*}
Then on the event $E_{N} \ce \{\norm{B}_{\infty} \leq 4^N - 1\}$ the process $V^N$ satisfies 
\begin{equation*}
    \abs{V^N_{t_i} - B_{t_i}} \leq \frac{1}{2^N}, \quad \forall N \in \N
\end{equation*}
and so we always have 
\begin{equation*}
    \norm{V^N - B}_\infty \leq \frac{1}{2^N} + \sup_{s, t \in [0, T]:\ |s - t| \leq 2^{-N}}|B_s - B_t|.
\end{equation*}
The right hand side converges to $0$ in $\P$ probability. Since $B$ remains a Brownian motion under each $\P^N$, $\P^N(E_N)$ converges to $1$. Therefore, for all $\varepsilon > 0$ we have
\begin{equation}\label{eq:V^NProperty}
    \lim_{N \to \infty}\P^N\pa{\norm{V^N - B}_\infty \leq \varepsilon}  = 1.
\end{equation}
With this choice of $V^N$, we have the following consistency property in the limit.

\begin{lemma}\label{Lemma:Consistency}
    For all bounded, uniformly continuous functions $h^0: \X \times \calP_1(\bar\Omega_1) \times \X_D \to \R$ and $h^{1}: \X^* \times \X \times \Gamma \times \R^2\to \R$, we have $$\barE\sqbra{h^0(B, \mM^\infty, D)h^{1}(X^\T, W^\infty, \upgamma^\infty, \calI)}=\barE \sqbra{h^0(B, \mM^\infty, D)\int_{\bar\Omega_1} h^1(x, w, q, \iota)d\mM^\infty(x, w, q, \iota)}.$$
\end{lemma}
\begin{proof}
    This equality holds at the discretization level by \eqref{eq:LiftedEnvironment}, namely
    \begin{align*}
        \E^N[h^0(V^N, \mM^N, D^N)h^{1}(X^\T, & W^N, \upgamma^N, \calI)]\\
        & = \E^N \sqbra{h^0(V^N, \mM^N, D^N)\int_{\bar\Omega_1} h^1(x, w, q, \iota)d\mM^N(x, w, q, \iota)}.
    \end{align*}
    By \eqref{eq:V^NProperty} and \eqref{eq:D^NProperty}, uniform continuity of the function $h^0$ allows us to swap $(V^N, D^N)$ with $(B, D)$ on both sides of the equality above without changing the limits, if they exist. Boundedness of $h_0, h_1$ and weak convergence enables us to take $N\to \infty$ and retain equality in the limit.
\end{proof}
Lemma \ref{Lemma:Consistency} says that $\mM^\infty$ is a version of conditional distribution of $(X^\T, W^\infty, \upgamma^\infty, \calI)$ given $(B, \mM^\infty, D)$. which implies the consistency of marginals $\mM^{\infty, x}$ and $\mM^{\infty, \upgamma}$. We want to carry the conditional joint laws because of the compatibility condition. 

\subsubsection{Optimality}
For each $N \in \N \cup \{\infty\}$, define $\mP^N \ce (\upmu^N, \upvartheta^N) =(\mM^{N, x}, \mM^{N, \upgamma})$ to be the first and third marginals of $\mM^N$. Define the state process corresponding to relaxed control any $\upgamma \in \bbGamma$ as in \eqref{eq:StateProcess_Relaxed} but in the environment $\P^N$ using $W^N$, namely
\begin{equation*}
    X^{N, \upgamma} \ce  \ind{t \geq \T}K_0/\mathscr{P} + 
    \intt\int_A a\upgamma(ds, da) + \sigma (W^N_{t \vee \T} - W^N_{\T}), \quad t \in [0, T].
\end{equation*}
Then in particular, $X^{N, \upgamma^N} = X^\T$. Recalling \eqref{eq:Objective_Relaxed} the objective function for relaxed controls, we now define for each $N$ the objective under the environment $\P^N$: $$J^N(\upgamma) \ce \E^{\barP^N}\sqbra{g(X^{N, \upgamma}_T, X^{N, \upgamma}_{\tau^*(\upmu^N)}, \tau^*(\upmu^N)) + \int_{\T \wedge T}^T \int_A f(s, X^{N, \upgamma}_s, b^P_s, \tau^*(\upmu^N), \qv{\rho, \upvartheta^N_s}, a)\upgamma(da,ds)}.$$

\begin{lemma}\label{Lemma:Optimality} $\lim_{N\to\infty}J^N(\upgamma^N)=J^\infty(\upgamma^\infty)$.
\end{lemma}
\begin{proof}
    Since the convergence in Lemma \ref{Lemma:TightnessLifted} is weak, we need to uniformly approximate $f, g$ by bounded functions. For $k \in \N$ and $x \in \R$, denote by $\underline{x}_k$ the projection of $x$ on $[-k, k]$. Define $f^k: [0, T] \times \R \times \R \times \R \times [0, T] \times \R \times A \to \R$ and $g: \Omega \times \R \times \R \times [0, T] \to \R$ by: 
    \begin{gather}
        \label{eq:fk}f^k(t, x, \mb, \upeta, \varrho, a) = \kappa(a) + \phi \underline{x}_k^2 - \underline{x}_k\pa{\underline{\mb}_k\ind{t < \upeta} + \varrho\ind{t \geq \upeta}},\\
        \label{eq:gk}g^k(x, y, \upeta) = c \underline{x}_k^2 + \beta_{\upeta}\underline{\gamma_{\upeta}}_k\underline{y}_k.
    \end{gather}
    Recall that $\gamma$ here is the bubble component defined in \eqref{eq:bubble_component}. Since we only care about $\gamma$ at the burst time, we can equivalently take $\gamma_t = \intt b^P_s ds$. Since the price impact functions $\kappa$ and $\rho$ are continuous, compactness of $A$ implies that for each $k \in \N$, there exists some $C_k > 0$ such that  $$\abs{g^k(X^\T_T, X^\T_{\tau^*(\upmu^N)}, \tau^*(\upmu^N))} + \abs{\int_{\T \wedge T}^T\int_A f^k(s, X^{\T}_s, b^P_s, \tau^*(\upmu^N), \qv{\rho, \upvartheta^N_s}, a)\upgamma(da,ds)} \leq C_k.$$
    For $N \in \N \cup \{\infty\}$ and $k \in \N$, define the approximated objective $J^{N, k}$ on $\bbGamma$ by $$J^{N, k}(\upgamma) \ce \E^{\P^N}\sqbra{g^k(X^\T_T, X^\T_{\tau^*(\upmu^N)}, \tau^*(\upmu^N)) + \int_{\T \wedge T}^T \int_A f^k(s, X^{\T}_s, b^P_s, \tau^*(\upmu^N), \qv{\rho, \upvartheta^N_s}, a)\upgamma(da,ds)}.$$
    Then weak convergence implies that $\lim_{N \to \infty}|J^{N, k}(\upgamma^N) - J^{\infty, k}(\upgamma^\infty)| = 0$. To shorten the notation, let $\tau^{N} \ce \tau^*(\upmu^N)$. Using \eqref{eq:Xestimate}, we have 
    \begin{align*}
        \sup_{N \in \N} \E^{\P^N}\sqbra{\sup_{t \in [0, T]}\abs{X^\T_t - \underline{X^\T_t}_k}^2} & = \sup_{N \in \N}\E^{\P^N}\sqbra{\sup_{t \in [0, T]}\abs{|X^\T_t| - k}^2\ind{|X^\T_t| > k}}\\
        & \leq \sup_{N \in \N}\E^{\P^N}\sqbra{\norm{X^\T}^2_{\X^*}\ind{\norm{X^\T}_{\X^*} > k}} \kto 0.
    \end{align*}
    Recall that $P$ has the same law under $\P^N$ for each $N$. Then similarly, Assumption \eqref{Assumption:B} and Proposition \ref{Prop:PriceSDEWellPosed} together imply
    \begin{align*}
        \sup_{N \in \N} \E^{\P^N}\sqbra{\sup_{t \in [0, T]}|b^P_t - \underline{b^P_t}_k|^2} & = \E^{\P^1}\sqbra{\sup_{t \in [0, T]}|b^P_t - \underline{b^P_t}_k|^2} \leq \E^{\P^1}\sqbra{\norm{b^P}_\infty\ind{\norm{b^P}_\infty > k}}\kto 0.\\
        \sup_{N \in \N} \E^{\P^N}\sqbra{\sup_{t \in [0, T]}|\gamma_t - \underline{\gamma_t}_k|^2} & = \E^{\P^1}\sqbra{\sup_{t \in [0, T]}|\gamma_t - \underline{\gamma_t}_k|^2} \leq \E^{\P^1}\sqbra{\norm{\gamma}^2_{\infty}\ind{\norm{\gamma}_\infty > k}}\\
        & \leq \E^{\P^1}\sqbra{T\norm{b^P}_\infty \ind{T\norm{b^P}_\infty > k}} \kto 0.
    \end{align*}
    These uniform integrability properties, along with Assumption \eqref{Assumption:B} and the separability condition in \ref{Remark:C1}, imply that there exists $C > 0$ such that  
    \begin{equation}\label{label:g-gk}
        \begin{gathered}
        \sup_{N \in \N}\E^{\P^N}\sqbra{|g - g^k|(X^\T_T, X^\T_{\tau^{N}}, \tau^{N})}  \leq c\sup_{N \in \N}\E^{\P^N}\sqbra{\abs{(X^\T_T)^2 - \underline{(X^\T_T)}_k^2}}\\ + C\sqrt{\sup_{N \in \N}\E^{\P^N}[|\gamma_{\tau^N} -\underline{\gamma_{\tau^N}}_k|^2]} + C\sqrt{\sup_{N \in \N}\E^{\P^N}[|X^\T_{\tau^N} - \underline{X^\T_{\tau^N}}_k|^2]} \ \kto 0,
     \end{gathered}
    \end{equation}
     and also
     \begin{equation}\label{label:f-fk}
    \begin{gathered}
     \sup_{N \in \N}\E^{\P^N}\Bigl[  \int_{\T \wedge T}^T\int_A |f - f^k|(s,  X^{\T}_s,  b^P_s, \tau^N, \qv{\rho, \upvartheta^N_s}, a)\upgamma^N(da,ds) \Bigr]\\ 
     \leq |\phi| \intT\sup_{N \in \N}\E^{\P^N}\sqbra{\abs{(X^\T_s)^2 - \underline{(X^\T_s)}_k^2}}ds
    + C\intT\sup_{N \in \N}\E^{\P^N}\sqbra{\abs{X^\T_s - \underline{(X^\T_s)}_k}^2}ds \\
    + \intT\sqrt{\sup_{N \in \N}\E^{\P^N}[|b^P_s -\underline{b^P_s}_k|^2]} + \sqrt{\sup_{N \in \N}\E^{\P^N}[|X^\T_s -\underline{X^\T_s}_k|^2]} \ ds \kto 0.
    \end{gathered}
    \end{equation}
Therefore, for any fixed $k \in \N$, 
\begin{align*}
\abs{J^N(\upgamma^N) - J^{\infty}(\upgamma^\infty)} & \leq \abs{J^N(\upgamma^N) - J^{N, k}(\upgamma^N)} + \abs{J^{N, k}(\upgamma^N) - J^{\infty, k}(\upgamma^\infty)} + \abs{J^{\infty, k}(\upgamma^\infty) - J^{\infty}(\upgamma^\infty)}.
\end{align*}
Taking limit $N\to \infty$ gives
\begin{equation*}
\lim_{N \to \infty}\abs{J^N(\upgamma^N) - J^{\infty}(\upgamma^\infty)} \leq \sup_{N \in \N}\abs{J^N(\upgamma^N) - J^{N, k}(\upgamma^N)} + \abs{J^{\infty, k}(\upgamma^\infty) - J^{\infty}(\upgamma^\infty)}.
\end{equation*}
Taking $k\to \infty$ on the right hand side and using \eqref{label:g-gk} and \eqref{label:f-fk} give the result.
    \end{proof}

Let $\beta \in \bbGamma$ be another $\mathbb{G}$-admissible relaxed strategy. Following the proof of Lemma \ref{Lemma:Optimality}, we also have $J^N(\beta) \Nto J(\beta)$. Remarks \ref{Remark:Optimality} and \ref{Remark:Relaxed_vs_Strict} together imply that $J^N(\upgamma^N) \leq J^N(\beta)$ for each $N$. Taking $N\to \infty$ on both sides we have $J^\infty(\upgamma^\infty) \leq J^\infty(\beta)$ for all $\beta \in \bbGamma$, so optimality is proved.

Now recall that for an $A$-valued control $\alpha \in \A^*$, we denote by $\upgamma(\alpha)$ its corresponding strict control in the space of relaxed controls, where each time marginal is the Dirac measure at $\alpha_t$. Using the optimality lemma above, we can in fact show that $\upgamma^\infty$ must be a strict control.
\begin{lemma}\label{Lemma:StrictControl}
    There is a version of $\upgamma^\infty$ that is $\mathbb{F}^{B, X^\T,  W^\infty, \calI, D}$-progressively measurable that is a strict control taking the form $\upgamma^\infty = \upgamma(\hat{\alpha}^\infty)$ for some $\hat{\alpha}^\infty \in \A^*$. 
\end{lemma}
\begin{proof}
    Define $\alpha^\infty_t \ce \int_A a\upgamma^\infty_t(da)$ for $t \in [0, T]$ and $\tilde{\upgamma} \ce \upgamma(\alpha^\infty)$. Then $\alpha^\infty \in \A^*$ and $\tilde{\upgamma} \in \bbGamma$ is a strict control. It is obvious that $\tilde{\upgamma}$ and $\upgamma^\infty$ both give rise to the same state process $X^\T$ according to \eqref{eq:StateProcess_Relaxed}. Using \emph{strict} convexity of $f$ in $a$ and Jensen's inequality, we have
    \begin{align*}
        J^\infty(\tilde{\upgamma}) & = \bar\E\sqbra{g(X^\T_T, X^\T_{\tau^*(\upmu^\infty)}, \tau^*(\upmu^\infty)) + \int_{\T \wedge T}^T \int_A f(s, X^{\T}_s, b^P_s, \tau^*(\upmu^\infty), \qv{\rho, \upvartheta^\infty_s}, a)\tilde{\upgamma}(da,ds)}\\
        & = \bar\E\sqbra{g(X^\T_T, X^\T_{\tau^*(\upmu^\infty)}, \tau^*(\upmu^\infty)) +\int_{\T \wedge T}^T f(s, X^{\T}_s, b^P_s, \tau^*(\upmu^\infty), \qv{\rho, \upvartheta^\infty_s}, \alpha^\infty_s)ds}\\
        & \leq \bar\E\sqbra{g(X^\T_T, X^\T_{\tau^*(\upmu^\infty)}, \tau^*(\upmu^\infty)) +\int_{\T \wedge T}^T \int_A f(s, X^{\T}_s, b^P_s, \tau^*(\upmu^\infty), \qv{\rho, \upvartheta^\infty_s}, a)\upgamma^\infty(da,ds)}\\
        & = J^\infty(\upgamma^\infty).
    \end{align*}
    The inequality is strict (which contradicts with optimality of $\upgamma^\infty$) unless $\upgamma^\infty = \upgamma(\alpha^\infty)$.  

    Lebesgue differentiation theorem allows us to define $\hat{\alpha}^\infty \in \A^*$ by $$\hat{\alpha}^\infty_t = \begin{cases}
        \limn n\int_{(t-1/n)+}^t \alpha^\infty_sds & \text{ if the limit exists}\\
        0 & \text{ otherwise}.
    \end{cases}$$
    Then $\bar\P \otimes dt$ almost surely, $\hat{\alpha}^\infty_t = \alpha^\infty_t$. Note that $\hat{\alpha}^\infty$ shares the same measurability with $\int_0^\cdot \alpha^\infty_s ds$, which by \eqref{eq:StateProcess} is $\mathbb{F}^{X^\T, W^\infty,B, \calI, D}$ measurable. 
\end{proof}
    A consequence of the lemma above is that we can drop either $\upgamma^\infty$ or $X^\T$
    from the definition of $\mathbb{G}$ and simply consider $\mathbb{G} = \mathbb{F}^{B, \mM^\infty, X^\T, W^\infty, \calI, D} = \mathbb{F}^{B, \mM^\infty, \upgamma^\infty, W^\infty, \calI, D}$. Moreover, both optimality and consistency still hold for $\upgamma^\infty = \upgamma(\hat{\alpha}^\infty)$. In fact, this is the case for every limit point of the sequence in Lemma \ref{Lemma:TightnessLifted}.
\subsubsection{Compatibility}
 Following Definition \ref{Def:weakMFGsolution_Relaxed}, we need to show that $\bar{\mathbb{F}} \ce \mathbb{F}^{\calI, \mM^\infty, W^\infty, B, D}$ is immersed in $\mathbb{G} = (\G_t)_{t\in [0, T]}$ defined above. We need to keep in mind that while $\calI$ is $\G_0$-measurable, it is not $\F^{\X^\T}_0$-measurable due to random entry, which is why we need to treat $\calI$ separately. 
\begin{lemma}\label{Lemma:Immersion}
    The filtration $\mathbb{G}$ is compatible with $(\calI, \mM^\infty, W^\infty, B, D)$.
\end{lemma}
\begin{proof}
    By Proposition \ref{Prop:Immersion}, it suffices to show that for all $t \in [0, T]$, $\F^{X^\T}_t$ is conditionally independent from $\bar\F_T = \F^{\calI, \mM^\infty, W^\infty, B, D}_T$ given $\bar\F_t = \F^{\calI, \mM^\infty, W^\infty, B, D}_t$. We follow the proof of \cite[Lemma 3.7]{CarmonaCommonNoise}. 
    
    Lemma \ref{Lemma:Consistency} implies that $W^\infty$ is a $\bar\P$-Brownian motion independent from $(B, \mM^\infty, D)$. Fix $t \in [0, T]$. Consider three bounded functions $\phi^m_t, \phi^w_{t+}$ and $\phi^1_t$ where $\phi^m_t: \calP_1(\bar\Omega_1) \to \R$ is $\F_t^{\mM^\infty}$ measurable, $\phi^w_{t+}: \X \to \R$ is $\sigma(W_s - W_t: s \in [t, T])$ measurable, and $\phi^1_t: \bar\Omega_1 \to \R$ is $\F^{\calI, X^\T, W^\infty}_t$ measurable. By Lemma \ref{Lemma:Consistency} and property of Brownian motion we have
    \begin{align*}
    \barE\sqbra{\phi_t^m(\mM^\infty)\int_{\bar\Omega_1}\phi^1_td\mM^\infty}\barE\sqbra{\phi^w_{t+}(W^\infty)} & = \barE\sqbra{\phi_t^m(\mM^\infty)\phi^1_t(X^\T, W^\infty, \upgamma^\infty, \calI)\phi^w_{t+}(W^\infty)}\\
        & =\barE\sqbra{\phi_t^m(\mM^\infty)\int_{\bar\Omega_1}\phi^w_{t+}(w)\phi^1_t(x, w, q, \iota)d\mM^\infty(x, w, q, \iota)}.
    \end{align*}
    Since this holds for all $\phi_t^m$, $\bar\P$ almost surely we have
    \begin{equation}\label{label:Lemma:Compatibility}
        \barE\sqbra{\phi^w_{t+}(W^\infty)}\E^{\mM^\infty}[\phi^1_t(X^\T, \upgamma^\infty, \calI, W^\infty)]= \E^{\mM^\infty}[\phi^w_{t+}(W^\infty)\phi^1_t(X^\T, \upgamma^\infty, \calI, W^\infty)]
    \end{equation}
    where by $\E^{\mM^\infty}[\phi(X^\T, \upgamma^\infty, \calI, W^\infty)]$ we mean the integral $\int_{\bar{\Omega}}\phi d\mM^\infty$ for $\phi : \bar\Omega_1 \to \R$. Note that this expectation is $\bar{\F}_t$-measurable if $\phi$ is $\F^{\calI, X^\T, W^\infty}_t$-measurable.
    
    Additionally, consider bounded functions $\phi^\iota, \phi^x_t, \varphi_t, \varphi_T$ where $\phi^\iota: \R^2 \to \R$ is Borel measurable, $\phi^x_t: \X^* \to \R$ is $\F^{X^\T}_t$ measurable, $\varphi_t$ and $\varphi_T$ are functions from $\X \times \calP_1(\bar\Omega_1) \times \X_D$ to $\R$ that are $\F^{B, \mM^\infty, D}_t$ and $\F^{B, \mM^\infty, D}_T$ measurable, respectively. Using \eqref{label:Lemma:Compatibility} and Lemma \ref{Lemma:Consistency}, we have
    \begin{equation}
    \begin{split}\label{label:ImmersionLemma}
    \barE[&\phi^x_t(X^\T) \varphi_T(B, \mM^\infty, D)\phi^\iota(\calI)\phi^w_{t+}(W^\infty)\phi_t^w(W^\infty)\varphi_t(B, \mM^\infty, D)]\\
    & = \barE\sqbra{\E^{\mM^\infty}[\phi^x_t(X^\T)\phi^\iota(\calI)\phi^w_{t+}(W^\infty)\phi_t^w(W^\infty)](\varphi_T\cdot\varphi_t)(B, \mM^\infty, D)}\\
    & = \barE\sqbra{\E^{\mM^\infty}[\phi^x_t(X^\T)\phi^\iota(\calI)\phi_t^w(W^\infty)](\varphi_T\cdot\varphi_t)(B, \mM^\infty, D)}\barE[\phi^w_{t+}(W^\infty)]\\
    & = \barE\biggl[\bar\E\sqbra{\phi^x_t(X^\T)\phi^\iota(\calI)\phi_t^w(W^\infty)| \bar\F_t}(\varphi_T\cdot\varphi_t)(B, \mM^\infty, D)\biggr]\barE[\phi^w_{t+}(W^\infty)]\\
    & = \barE\biggl[\bar\E\sqbra{\phi^x_t(X^\T)\phi^\iota(\calI)\phi_t^w(W^\infty)| \bar\F_t}\varphi_t(B, \mM^\infty, D)\barE\sqbra{\varphi_T(B, \mM^\infty, D)|\bar{\F}_t}\biggr]\barE[\phi^w_{t+}(W^\infty)]\\
    & = \barE\sqbra{\barE[\phi_t^x(X^\T)|\bar\F_t]\barE\sqbra{\varphi_T(B, \mM^\infty, D)|\bar\F_t}\varphi_t(B, \mM^\infty, D)\phi^\iota(\calI)\phi_t^w(W^\infty)\phi^w_{t+}(W^\infty)},
    \end{split}
    \end{equation}
    where the last equality follows from the independence of $\phi^w_{t+}(W^\infty)$ and $\bar\F_t$. Since $\phi^\iota$ and $\phi_t^w$ are arbitrary, we can replace them with bounded $\phi^{\iota}\cdot\varphi^{\iota}$ and $\phi^w_t \cdot \varphi_t^w$, each with the same corresponding mesurability requirements. Then by definition of conditional expectation we have 
    \begin{align*}
    \barE\bigl[\phi^x_t( X^\T) &\varphi_T(B, \mM^\infty, D)\phi^\iota(\calI)\phi^w_{t+}(W^\infty)\phi_t^w(W^\infty)| \bar\F_t\bigr]\\
    & = \barE\sqbra{\phi^x_t( X^\T)| \bar\F_t}\barE\sqbra{\varphi_T(B, \mM^\infty, D)\phi^\iota(\calI)\phi^w_{t+}(W^\infty)\phi_t^w(W^\infty) | \bar\F_t}.
    \end{align*}
    We conclude by noting that $\F^{W^\infty}_T$ is generated by $\phi^w_{t+}(W^\infty)\phi_t^w(W^\infty)$ with arbitrary $\phi^w_{t+}$ and $\phi^w_t$.
\end{proof}
We have then finished the proof of Theorem \ref{Theorem:Existence_Weak}.

\section{Strong Control and Separability by Burst}\label{Section:StrongControl}
\subsection{Strong Control in Original Environment}
Recall from Lemma \ref{Lemma:StrictControl} that the weak control found in the previous section is in fact a strict control $\upgamma(\alpha^\infty)$, and $\alpha^\infty$ is $\mathbb{F}^{B, X^\T, W^\infty, \mM^\infty, \calI, D}$ progressive. In order to obtain an equilibrium with strong control, we will show that $\alpha^\infty$ is $\mathbb{F}^{B, W^\infty, \mP^\infty, \calI, D}$ measurable after bringing the lifted environment $\mM^\infty$ back to the ``original'' environment $\mP^\infty = (\upmu^\infty, \upvartheta^\infty) = (\mM^{\infty, x}, \mM^{\infty, \upgamma})$. 
\subsubsection{Back to Original Environment}
The reason for lifting the environment is solely for the proof of the compatibility lemma \ref{Lemma:Immersion}, in particular the first and third equality in \eqref{label:ImmersionLemma}. Recall from \eqref{eq:LiftedEnvironment} that we took $\mM^N$ to be the joint conditional law of $(X^\T, W^N, \upgamma^N, \calI)$ given $(B, D)$ under $\P^N$. We did this to ease the notation in the consistency and compatibility lemmas. Notice that we did not need the full fledged joint law in deriving \eqref{label:ImmersionLemma}, but only the product of the marginals. This implies that for fixed $N \in \N$, we could alternatively define for each $(\beta, \upeta) \in \X \times \R_+$:
$$\widetilde{\mM}^N(\beta, \upeta)  \ce \sum_{k = 1}^{|\calV_N|}\ind{(\beta, D(\upeta)) \in A_k}\L^{\alpha^N, k}(X^\T) \otimes \L^{\alpha^N, k}(W^N)\otimes \L^{\alpha^N, k}(\upgamma^N) \otimes \L^{\alpha^N, k}(\calI).$$
This version still carries the necessary inputs $\mP^N$ to the BSDE \eqref{eq:BSDE_Main} as its first and third marginals, and tightness of $\P^N\circ (\widetilde{\mM}^N)^{-1}$ follows immediately from that of $\P^N\circ (\mM^N)^{-1}$ in Lemma \ref{Lemma:TightnessLifted}. Then we take a limit point $\bar\P \ce \P^\infty \circ (B, \tau,\widetilde{\mM}^\infty, X^\T, W^\infty, \upgamma^\infty, \calI)^{-1}$ and follow the same argument in Remark \ref{Remark:Candidate} to work on the canonical space. In particular, $\widetilde{\mM}^\infty$ is the canonical process on $\calP_1(\bar{\Omega}_1)$. Following the argument in Lemma \ref{Lemma:Consistency}, the fixed point property for each $N \in \N$ now leads to a weaker consistency in the limit. Namely, for all bounded, uniformly continuous, $\R$-valued functions $h^0, h^1_x, h^1_w, h^1_\upgamma, h^1_\iota$ with respective domains $\X \times \calP_1(\bar{\Omega}_1) \times \X_D, \X^*, \X, \Gamma, \R^2$, we have
\begin{equation}\label{eq:Consistency_Product}
\begin{split}
    \barE \biggl[&h^0(B, \widetilde{\mM}^\infty, D)\int_{\X^*} h^1_x(x)d\widetilde{\mM}^{\infty}(x)\int_{\X}h^1_w(w)d\widetilde{\mM}^{\infty}(w)\int_{\Gamma}h^1_{\upgamma}(q)d\widetilde{\mM}^{\infty}(q)\int_{\R^2}h^1_\iota(\iota)d\widetilde{\mM}^{\infty}(\iota)\biggr]\\
    & = \barE[h^0(B, \widetilde{\mM}^\infty, D)h^{1}_x(X^\T)h^1_w(W^\infty)h^1_{\upgamma}(\upgamma^\infty)h^1_\iota(\calI)]\\
    & = \barE \sqbra{h^0(B, \widetilde{\mM}^\infty, D)\int_{\bar\Omega_1} h^1_x(x)h^1_w(w)h^1_{\upgamma}(q)h^1_\iota(\iota)d\widetilde{\mM}^\infty(x, w, q, \iota)},
\end{split}
\end{equation}
which results from taking $N \to \infty$ of the following equalities by construction
\begin{equation*}
\begin{split}
    \barE \biggl[&h^0(V^N, \widetilde{\mM}^N, D^N)\int_{\X^*} h^1_x(x)d\widetilde{\mM}^{N}(x)\int_{\X}h^1_w(w)d\widetilde{\mM}^{N}(w)\int_{\Gamma}h^1_{\upgamma}(q)d\widetilde{\mM}^{N}(q)\int_{\R^2}h^1_\iota(\iota)d\widetilde{\mM}^{N}(\iota)\biggr ]\\
    & = \barE[h^0(V^N, \widetilde{\mM}^N, D^N)h^{1}_x(X^\T)h^1_w(W^N)h^1_{\upgamma}(\upgamma^N)h^1_\iota(\calI)]\\
    & = \barE \sqbra{h^0(V^N, \widetilde{\mM}^N, D^ N)\int_{\bar\Omega_1} h^1_x(x)h^1_w(w)h^1_{\upgamma}(q)h^1_\iota(\iota)d\widetilde{\mM}^N(x, w, q, \iota)}.
\end{split}
\end{equation*}

This is also sufficient for the consistency requirement in Definition \ref{Def:weakMFGsolution_Relaxed}. Similarly, with $\widetilde{\mM}^\infty$ the equality \eqref{label:Lemma:Compatibility} holds only for $\phi^1_t$ taking the form of a product, separable in each coordinate. This weaker property, however, is sufficient for \eqref{label:ImmersionLemma} and hence the compatibility requirement. Since the optimality property only depends on the marginals and thus is not influenced, we can replace $\mM^\infty$ with $\widetilde{\mM}^\infty$ in the final filtration $\mathbb{G} = \mathbb{F}^{B, X^\T, W^\infty, \widetilde{\mM}^\infty, \calI, D}$.

Note from \eqref{eq:Consistency_Product} that $\barP$ almost surely, $\widetilde{\mM}^\infty$ is a product measure of its four marginals by uniqueness of measures on the product space. More importantly, its second and fourth are almost surely the Wiener measure and $(\nu_K \otimes \nu_p)$, respectively, since for each $N \in \N$, $(W^N, \calI, B, D)$ are mutually independent under $\P^N$. Being complete, the filtration $\mathbb{F}^{\mP^{\infty}}$ in the original environment coincides with $\mathbb{F}^{\widetilde{\mM}^\infty}$ from the lifted environment. Therefore, we can equivalently take $\mathbb{G} = \mathbb{F}^{B, X^\T, W^{\infty}, \mP^\infty, \calI, D}$, and the compatibility condition reads that $\bar{\mathbb{F}} \ce \mathbb{F}^{B, W^{\infty}, \mP^\infty, \calI, D}$ is immersed in $\mathbb{G}$.

\subsubsection{Strong Control via Optional Projection}
To further strengthen the measurability property of $\upgamma$ from $\mathbb{G}$ to $\bar{\mathbb{F}}$, we follow the proof of \cite[Proposition 4.4]{CarmonaCommonNoise}. Recall that the state equation \eqref{eq:StateProcess} with $W^\infty$ as the Brownian motion is satisfied by $X^\T$ and $\hat{\alpha}^\infty$. By optional projection we can find $\bar{\mathbb{F}}$-optional processes $\bar{X}^\T$ and $\bar{\alpha}$ such that for any finite $\bar{\mathbb{F}}$-stopping time $\rho$: 
\begin{equation}\label{eq:OptionalProjection_1}
    \bar{X}^{\T}_{\rho} \ce \barE[X^\T_{\rho}|\bar{\F}_{\rho}], \quad \bar{\alpha}_{\rho} \ce \barE[\hat{\alpha}^\infty_{\rho} | \bar{\F}_{\rho}], \quad \as
\end{equation}
Since $\bar{\mathbb{F}}$ is immersed in $\mathbb{G}$, Proposition \ref{Prop:Immersion} implies that for each $0 \leq s \leq t \leq T$, 
\begin{equation}\label{eq:OptionalProjection_2}
    \bar{X}^{\T}_{s} = \barE[X^\T_{s}|\bar{\F}_t], \quad \bar{\alpha}_{s} = \barE[\hat{\alpha}^\infty_{s} | \bar{\F}_t], \quad \as
\end{equation}
Using Fubini's theorem for conditional expectation along with \eqref{eq:OptionalProjection_2} on \eqref{eq:StateProcess}, we can replace $\bar{X}^{\T}$ by a modification such that $\bar{\P}$ almost surely 
\begin{equation*}
    \bar{X}^{\T}_t = \ind{t \geq \T}K_0/\mathscr{P} + 
    \intt \bar{\alpha}_s ds + \sigma (W^\infty_{t \vee \T} - W^\infty_{\T}), \quad t \in [0, T].
\end{equation*}
Notice that given $\mP^\infty = (\widetilde{\mM}^{\infty, x}, \widetilde{\mM}^{\infty, \upgamma}) \equiv (\upmu^\infty, \upvartheta^{\infty})$, the bubble burst time $\tau^*(\upmu^\infty)$ is a $\bar{\mathbb{F}}$-stopping time, and the bubble component $\gamma_{\tau^*(\upmu^\infty)}$ is $\bar{\F}_{\tau^*(\upmu^\infty)}$ measurable. Recall also from \ref{Remark:C1} that the running cost $f$ depends on $\tau^*(\upmu^\infty)$ only through $D$, which is $\bar{\mathbb{F}}$-adapted.
Then by conditional Jensen's inequality, Remark \ref{Remark:C3} and \eqref{eq:OptionalProjection_1},
\begin{align*}
    J^{\upmu^\infty, \upvartheta^\infty}(\hat\alpha^\infty) & = \barE\sqbra{\int_{\T\wedge T}^Tf(s, X^{\T}_s, b^P_s, \tau^*(\upmu^\infty), \qv{\rho, \upvartheta_s^\infty}, \hat{\alpha}_s)ds + c(X^\T_T)^2 +\beta_{\tau^*(\upmu^\infty)}\gamma_{\tau^*(\upmu^\infty)}X^{\T}_{\tau^*(\upmu^\infty)}}\\
    & = \barE\sqbra{\int_{\T\wedge T}^T\barE\sqbra{f(s, X^{\T}_s, b^P_s, \tau^*(\upmu^\infty), \qv{\rho, \upvartheta_s^\infty}, \hat{\alpha}_s)|\bar{\F}_s}ds}\\
    & \qquad \qquad + c\barE\sqbra{\barE\sqbra{|X^{\T}_T|^2 | \bar{\F}_T}} + \barE\sqbra{\beta_{\tau^*(\upmu^\infty)}\gamma_{\tau^*(\upmu^\infty)}\barE\sqbra{X^{\T}_{\tau^*(\upmu^\infty)} | \bar{\F}_{\tau^*(\upmu^\infty)}}}\\
    & \geq \barE\sqbra{\int_{\T\wedge T}^Tf(s, \bar{X}^{\T}_s, b^P_s, \tau^*(\upmu^\infty), \qv{\rho, \upvartheta_s^\infty}, \bar{\alpha}_s)ds + c(\bar{X}^\T_T)^2 +\beta_{\tau^*(\upmu^\infty)}\gamma_{\tau^*(\upmu^\infty)}\bar{X}^{\T}_{\tau^*(\upmu^\infty)}}\\
    & = J^{\upmu^\infty, \upvartheta^\infty}(\bar{\alpha}).
\end{align*}
By strict convexity of $f$ in $(x, a)$, the inequality is strict unless $\hat{\alpha}^\infty$ and $X^\T$ are both already $\bar{\mathbb{F}}$ adapted. Strict inequality would lead to a contradiction to optimality of $\hat{\alpha}^\infty$ among $\mathbb{G}$-progressive controls, since $\bar{\alpha}$ is $\bar{\mathbb{F}}$-optional, hence also $\mathbb{G}$-progressive. 

\subsubsection{Exogenous Burst Time as Totally Inaccessible Stopping Time}
The section above implies that we can take $\mathbb{G} = \bar{\mathbb{F}} = \mathbb{F}^{\calI, B, W^\infty, D, \mP^\infty}$ to begin with. This concludes the proof for the existence statement of Theorem \ref{Theorem:MFGExistence}. We now mention a desired feature for the bubble model as a corollary.
\begin{corollary}
    The exogenous burst time $\tau$ is a $\mathbb{F}^{\calI, B, W^\infty, D, \mP^\infty}$-totally inaccessible stopping time.
\end{corollary}
\begin{proof}
    In light of Remark \ref{Remark:Inaccessible} and Assumption \ref{Assumption:E2}, it suffices to remark that $\tau$ is independent from $(\calI, B, W^\infty, \mP^\infty)$, which follows from the independence between $\tau$ and $(\calI, B, W^N, \mP^N)$.
\end{proof}

\section{Concluding Remarks}
In this paper we proposed a more realistic extension of the bubble riding game introduced in \cite{TangpiWang22}.
In contrast to \cite{TangpiWang22} where agents were assumed to enter the game at independent and identically distributed times on an awareness window $[0,\eta]$, here we allow players to enter the game when the price trajectory of the bubble asset reaches a given threshold.
We also allow the initial inventory to depend on the initial (cash) investment and the price level at time of entry. 
Due to these improvements on the model, the resulting MFG in the $N\to\infty$ limit is one with common noise in addition to non-standard features such as random entry times, interaction through the controls and possible jump of the state processes.
Because the coefficients of the game do not satisfy the usual monotonicity conditions assumed in common noise MFG theory, we have to settle for existence of equilibria in a suitable weak form (see Definition \ref{Def:weakMFGsolution}).
In short, the weaker, more realistic model assumptions made in the present paper result in weak, abstract equilibrium strategies whereas the stronger model assumptions made in \cite{TangpiWang22} result in stronger equilibrium strategies that can be numerically simulated thus providing interesting economical insights.

\appendix
\section{Two Auxiliary results}
For a c\`{a}dl\`{a}g process $Y$, denote by $M^Y_t = \sup_{0 \leq s \leq t}Y_s$ its running maximum. Recall from \eqref{eq:N_in} and the price dynamics of the N-player game that the bubble trend function $b$ naturally depends on $F_p(M^P_t)$, which is not Lipschitz in $M^P$. In general, the dynamics of asset price in the bubble phase is not well-posed. However, as the bubble is fueled by players' entry, $b$ should be increasing in $F_p(M^P_t)$, hence also increasing in $M^P_t$ at each time $t \in [0, T]$ since $F_p$ is a CDF. This monotonicity property of the path-dependent SDE \eqref{eq:PriceDynamics+} restores unique solvability.
\begin{proposition}\label{Prop:PriceSDEWellPosed}
The following path-dependent SDE
\begin{equation}\label{eq:MFGPriceDynamics}
    X_t = x + \intt \tilde{b}(s, M^X_s, X_s)ds + \sigma_0 B_t
\end{equation}
has a unique strong solution satisfying $\E[\norm{X}_\infty^2] < \infty$ if for each fixed $t\in [0,T]$:
\begin{enumerate}
    \item There exists $C > 0$ such that for all $\bx \in C([0, T]; \R)$: $$\abs{\tilde{b}\pa{t, M^{\bx}_t, \bx_t}} \leq C\pa{1 +  M^{|\bx|}_t}.$$
    \item $\tilde{b}(t, \cdot, \cdot)$ is increasing (not necessarily strictly) in each argument.
\end{enumerate}
\end{proposition}

\begin{proof}
We adapt the proof of \cite[Theorem 4.1]{belfadli2009}. The first condition guarantees a weak solution satisfying the integrability condition that is unique in law (see \cite[Proposition 5.3.6 and Remark 5.3.8]{bookKaratzasShreve91}). By the well-known result of \citet{yamada71}, we only need to show pathwise uniqueness. Suppose $X$ and $Y$ are two solutions on the same probability space with respect to the same Brownian motion $B$. Observing that $X - Y$ is absolutely continuous, by Tanaka's formula we get
\begin{equation}\label{smalleq:XmaxY}
\begin{split}
X_t \vee Y_t & = Y_t + (X_t - Y_t)^+ = Y_t + \intt \ind{X_s > Y_s}d(X_s - Y_s) \\
& = x + \sigma_0 B_t + \intt \ind{X_s > Y_s}\tilde{b}(s, M^X_s, X_s)ds + \intt \ind{X_s \leq Y_s}\tilde{b}(s, M^Y_s, Y_s)ds.\\
Y_t \vee X_t & =  x + \sigma_0 B_t + \intt \ind{Y_s > X_s}\tilde{b}(s, M^Y_s, Y_s)ds + \intt \ind{Y_s \leq X_s}\tilde{b}(s, M^X_s, X_s)ds.
\end{split}
\end{equation}
We can equate the above expressions for all $t$, implying that for almost every $t$ we have 
\begin{equation}\label{smalleq:samedriftXYae}
\ind{X_t = Y_t} \pa{\tilde{b}(t, M^Y_t, Y_t) - \tilde{b}(t, M^X_t, X_t)} = 0.
\end{equation}
We now show that if $X_s > Y_s$, then $M^X_s \geq M^Y_s$. Define $$s_0 \ce \sup\{u \in [0, s]: X_u = Y_u\}.$$ 
The case is trivial if $s_0 = 0$.

On the event $\{s_0 > 0\}$, continuity of $X$ and $Y$ implies that $X_t > Y_t$ for all $t \in (s_0, s]$.  Suppose $M^X_s < M^Y_s$, then there must exist $s^* \in [0, s_0)$ where $Y_{s^*} = M^Y_s > M^X_s \geq X_{s^*}$. Then define $$s_1 \ce \inf\{u \in [s^*,s_0]: X_u = Y_u\}.$$
By continuity again, $Y_t > X_t$ for all $t \in [s^*, s_1)$. By definition of $s^*$, we must also have $M^Y_{t} > M^X_{t}$ for all $t \in [s^*, s_1)$. Monotonicity of $\tilde{b}$ leads to a contradiction $$0 > X_{s^*} - Y_{s^*} = \int_{s^*}^{s_1}\tilde{b}(t, M^Y_t, Y_t) - \tilde{b}(t, M^X_t, X_t)dt\geq 0.$$
Therefore, $M^X_s \geq M^Y_s $ and in particular, $M^{X\vee Y}_s = M^{X}_s$.
We can then rewrite \eqref{smalleq:XmaxY} as    
\begin{align*}
X_t \vee Y_t & = x + \sigma_0 B_t + \intt \tilde{b}(s, M^{X \vee Y}_s, X_s \vee Y_s)ds \\
&+ \intt \ind{X_s = Y_s}\pa{\tilde{b}(s, M^{Y}_s, Y_s) - \tilde{b}(s, M_s^{X \vee Y},  X_s \vee Y_s)}ds\\
& = x + \sigma_0 B_t + \intt \tilde{b}(s, M^{X \vee Y}_s, X_s \vee Y_s)ds \\
& + \intt \ind{\{X_s = Y_s\} \cap \{M^{X}_s > M^{Y}_s\}}\pa{\tilde{b}(s, M^{Y}_s, Y_s) - \tilde{b}(s, M_s^{X \vee Y},  X_s \vee Y_s)}ds.
\end{align*}
where the last line vanishes by \eqref{smalleq:samedriftXYae}. Therefore, $X \vee Y$ also satisfies \eqref{eq:MFGPriceDynamics}. Similarly, one can show $X \wedge Y$ is also a solution. Then by uniqueness of law, we have $\E[|X - Y|] = \E[X\vee Y - X\wedge Y] = 0$ which leads to pathwise-uniqueness and completes the proof. The integrability property easily follows from Gr\"{o}nwall's inequality.
\end{proof}
The following measure theoretic result is probably well known. We give a proof since we could not find a directly citable reference.

\begin{lemma}\label{Lemma:Caratheodory}
Let $(S, \Sigma, \mu)$ be a complete measurable space. A function $f: S \times \R \to \R$ is jointly measurable if for all $x \in \R$:
\begin{enumerate}
    \item $f(\cdot, x)$ is measurable.
    \item $f(\cdot, x_n)$ converges to $f(\cdot, x)$ in $\mu$-measure for any increasing sequence $x_n \uparrow x$. 
\end{enumerate}
\end{lemma}
\begin{proof}

    First let $E \subseteq \R$ be any closed set and let $X = \{x_m\}_{m \geq 1}$ be a countable, dense subset of $\R$. For $\varepsilon> 0$, denote by $\mathcal{O}_{\epsilon}(E)$ the open set $\{x \in \R: \inf_{e \in E}|x - e| < \varepsilon\}$. We claim that for $\mu$-almost every $s \in S$ and any $x \in \R$, $f(s, x) \in E$ if and only if for each $n \in \N$, there is $x_m \in X \cap (x - \frac{1}{n}, x]$ such that $f(s, x_m) \in \mathcal{O}_{\frac{1}{n}}(E)$. Note that we can always approximate any $x \in \R$ by an increasing sequence $\{x_{m_k}\}_{k \geq 1}$ with elements in $X$ such that the functions $f(\cdot, x_{m_k})$ converge $\mu$-almost everywhere to $f(\cdot, x)$. The claim follows almost immediately. Denoting by $f^{-1}$ the preimage of $f$, joint measurability is proved by writing 
    $$f^{-1}(E) = \bigcap_{n = 1}^\infty \bigcup_{m = 1}^{\infty}\cbra{s \in S :f(s, x_m) \in \mathcal{O}_{\frac{1}{n}}(E)} \times \left[x_m, x_m + \frac{1}{n}\right).$$
\end{proof}

\bibliographystyle{plainnat}
\bibliography{references}

\vspace{1cm}

\noindent Princeton University\\
Operations Research and Financial Engineering\\
Email address: shichun.wang@princeton.edu

\vspace{.5cm}

\noindent Princeton University\\
Operations Research and Financial Engineering\\
Bendheim Center for Finance \\
Email address: ludovic.tangpi@princeton.edu
\end{document}